\documentclass[preprint]{sigplanconf}

\newcommand{\techRep}{true} \newcommand{\iftechrep}{\ifthenelse{\equal{\techRep}{true}}}

\usepackage{amsfonts,amstext,amsmath,amssymb,stmaryrd,mathrsfs,calc,amsthm,xspace,mathtools}
\usepackage{bm}\usepackage{comment}
\usepackage{framed}
\usepackage{graphics}
\usepackage{enumerate}
\usepackage{xspace,tabularx,multirow}

\usepackage{todonotes}
\usepackage{mathtools}
\usepackage{complexity}

\usepackage{thm-restate}

\usepackage{color,soul}

\usepackage{macros}
\usepackage{hyperref}

\usepackage{lipsum}

\usepackage{tikz}
\usetikzlibrary{arrows,calc,automata}
\tikzset{LMC style/.style={>=stealth',every edge/.append style={thick},every state/.style={minimum size=10,inner sep=0}}}

\begin{document}

\setlength{\pdfpageheight}{\paperheight}
\setlength{\pdfpagewidth}{\paperwidth}

\conferenceinfo{CONF 'yy}{Month d--d, 20yy, City, ST, Country} 
\copyrightyear{20yy} 
\copyrightdata{978-1-nnnn-nnnn-n/yy/mm} 
\doi{nnnnnnn.nnnnnnn}

\title{Efficient Quantile Computation in Markov Chains via Counting
  Problems for Parikh Images}

\authorinfo{Christoph Haase\thanks{Supported by Labex Digicosme, 
    Univ.\ Paris-Saclay, project VERICONISS.}}
{LSV, CNRS \& ENS Cachan\\Universit\'e Paris-Saclay, France}
{haase@lsv.ens-cachan.fr}

\authorinfo{Stefan Kiefer}{Department of Computer Science\\University of Oxford, UK}
{stekie@cs.ox.ac.uk}

\authorinfo{Markus Lohrey}{Department f\"{u}r Elektrotechnik und Informatik\\
  Universit\"{a}t Siegen, Germany}{lohrey@eti.uni-siegen.de}

\maketitle

\begin{abstract}
A cost Markov chain is a Markov chain whose transitions are labelled with non-negative integer costs. A fundamental problem on this model, with applications in the verification of stochastic systems, is to compute information about the distribution of the total cost accumulated in a run. This includes the probability of large total costs, the median cost, and other quantiles. While expectations can be computed in polynomial time, previous work has demonstrated that the computation of cost quantiles is harder but can be done in PSPACE. In this paper we show that cost quantiles in cost Markov chains can be computed in the counting hierarchy, thus providing evidence that computing those quantiles is likely not PSPACE-hard. We obtain this result by exhibiting a tight link to a problem in formal language theory: counting the number of words that are both accepted by a given automaton and have a given Parikh image. Motivated by this link, we comprehensively investigate the complexity of the latter problem. Among other techniques, we rely on the so-called BEST theorem for efficiently computing the number of Eulerian circuits in a directed graph.
\end{abstract}

\section{Introduction}\label{sec-introduction}
\makeatletter{}
Markov chains are an established mathematical model that allows for
reasoning about systems whose behaviour is subject to stochastic
uncertainties. A Markov chain comprises a set of states with a
transition function that assigns to every state a probability
distribution over the set of successor states. A typical problem about a
given Markov chain is computing the
probability with which a designated target state is reached starting
from an initial state. This probability is computable in polynomial
time~\cite{BK08} for explicitly given Markov chains. Polynomial-time
decidability carries over to properties specified in
PCTL~\cite{BK08}, a stochastic extension of the branching-time logic
CTL.
Probabilistic model checkers such as PRISM~\cite{KNP11} and
MRMC~\cite{KZHHJ11} can efficiently reason about such properties on
large Markov chains in practice.

In order to gain flexibility for modelling systems, a
natural generalisation is to extend transitions of Markov chains with
integer weights. Those weights can model the cost that is incurred or the time that elapses
when moving from one state to another. Beyond reachability probabilities,
one can then, for instance, ask for the expected value of the accumulated weight along the paths reaching a target state.
Such an expectation can be computed in polynomial time~\cite{BK08}, 
but it provides little information on the \emph{guaranteed behaviour} of a system.
\tikzset{initial text={}}
\begin{figure}
\begin{center}
  \begin{tikzpicture}[xscale=1.7,yscale=0.85,LMC style]
    \node[state, initial] (a) at (0,0) {$s$};
    \node[state, accepting] (b) at (3,0) {$t$};
    \node[state] (c) at (1.5,1) {$u$};

    \path[->] (a) edge node[below] {$t_1:20'/0.9$} (b);
    \path[->, bend left] (a) edge node[above] {$t_2:15'/0.1$} (c);
    \path[->] (c) edge [loop above,looseness=60] node {$t_2:5'/0.2$} (c);
    \path[->, bend left] (c) edge node[above] {$t_3:10'/0.8$} (b);
  \end{tikzpicture}
  \label{fig:airport-security}
  \caption{Simplified model of an airport security procedure.}
\end{center}
\end{figure}
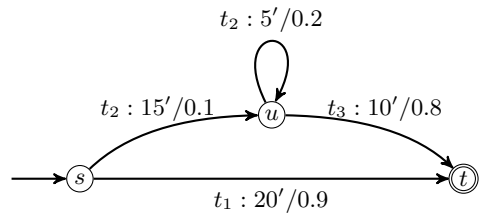
For an example, consider the simplified model of an airport security
procedure illustrated in Figure~\ref{fig:airport-security}.
Starting in state $s$ a passenger is either (with probability 0.9) routed through the standard security check, which takes 20min, or (with probability~0.1) through the extended security check whose first stage takes 15min.
After the first stage, the passenger reaches state~$u$ where she is, with probability 0.2, subject to repeated additional security screenings, each of which takes 5min.
Once completed, it takes a passenger another 10min to complete the extended
security check and to reach the airport gate. It can easily be
verified that the expected value of the time required to reach state
$t$ is $\approx 21.3$min.
Suppose an airport operator wants to find
out if it can guarantee that 99.999\% of its customers clear the
security check and reach the gate within 30min. Knowing the expected
time does not suffice to answer this question. In
fact, a simple calculation shows that only 99.996\% of the customers
complete the security check within 30min.

The quantile problem considered in that example is an instance of
the more general \emph{cost problem}: Given a Markov chain whose
transitions are labelled with non-negative integers and which has a
designated target state $t$ that is almost surely reached (called a
\emph{cost chain} in the following), a probability threshold $\tau$
and a Boolean combination of linear inequalities over one variable
$\varphi(x)$, the cost problem asks whether the accumulated
probabilities of paths achieving a value consistent with $\varphi$
when reaching $t$ is at least $\tau$. It has been shown in~\cite{HK15}
by the first two authors that the cost problem can be decided
in \PSPACE. The starting point of this paper is the question left open
in~\cite{HK15} whether this \PSPACE-upper bound can be improved.

Our first contribution is to answer this question positively:
we show that the cost problem belongs to the counting hierarchy
(\CH). The counting hierarchy is defined similarly to the
polynomial-time hierarchy using counting quantifiers, see
\cite{AllenderW90} or Section~\ref{ssec:complexity} for more details. It is contained in \PSPACE\ and
this inclusion is believed to be strict. In recent years, several
numerical problems, for which only \PSPACE\ upper bounds had been known,
have been shown to be in \CH. Two of the most important
and fundamental problems of this kind are \PosSLP\ and \BitSLP:
\PosSLP\ is the problem whether a given arithmetic circuit over the
operations $+$, $-$ and $\times$ evaluates to a positive number, and
\BitSLP\ asks whether a certain bit of the computed number is equal to
$1$. Note that an arithmetic circuit with $n$ gates can evaluate to a
number in the order of $2^{2^n}$; hence the number of output bits can be
exponential and a certain bit of the output number can be specified
with polynomially many bits. In addition to the \PSPACE\ upper bound,
it has been shown in~\cite{HK15} that the cost problem is hard for
both \PosSLP\ and \PP\ (probabilistic polynomial time).

In order to show that the cost problem belongs to \CH, we identify a
counting problem for certain words accepted by a deterministic
finite-state automaton as the core underlying problem and show its
membership in \CH. Relating to our previous example in
Figure~\ref{fig:airport-security}, observe that any path that reaches
the state $t$ induces a function $\vec{p}$ mapping every transition
$t_i$ to the number of times $t_i$ is traversed along this path. In
particular, given such a function $\vec{p}$ we can easily check
whether it exceeds a certain time budget and compute the probability
of a path with the induced function $\vec{p}$. Thus, viewing a cost
chain as a deterministic finite-state automaton whose edges are
labelled by alphabet symbols $t_i$, this observation gives rise to the
following two counting problems for words that are defined analogously
to \PosSLP\ and \BitSLP: For a finite-state automaton $\A$ over a
finite alphabet $\Sigma$ and a Parikh vector $\vec{p}$ (i.e., a mapping
from $\Sigma$ to $\mathbb{N}$) we denote by $N(\A, \vec{p})$ the
number of words accepted by $\A$ whose Parikh image is $\vec{p}$. Then
\BitP\ is the problem of computing a certain bit of the number $N(\A,
\vec{p})$ for a given finite-state automaton $\A$ and a Parikh vector
$\vec{p}$ encoded in binary. Further, \PIC\ is the problem of checking
whether $N(\A, \vec{p}) > N(\B, \vec{p})$ for two given automata $\A$
and $\B$ (over the same alphabet) and a Parikh vector $\vec{p}$ encoded
in binary. We prove that \BitP\ and \PIC\ both belong to the counting
hierarchy if the input automata are deterministic. The main ingredient
of our proof is the so-called BEST theorem which gives a formula for
the number of Eulerian circuits in a directed graph. From the BEST
theorem we derive a formula for the number $N(\A, \vec{p})$. In
addition, based on techniques introduced
in~\cite{ABKB09,HAMB02}, we develop a toolbox for showing membership of numerical
problems in the counting hierarchy. This  enables us to evaluate the formula for $N(\A, \vec{p})$ in
the counting hierarchy. We then reduce the cost problem for Markov
chains to the problem of comparing two numbers, one of which involves
$N(\A, \vec{p})$ for a certain deterministic finite state automaton
$\A$ and Parikh vector $\vec{p}$ encoded in binary. Since single bits
of this number can be computed in \CH, we finally obtain our main
result: The cost problem belongs to \CH. To the best of our knowledge,
the cost problem and \PIC\ (for DFA with Parikh vectors encoded in
binary) are the only known natural problems, besides \BitSLP, which are (i) \textsc{PosSLP}-hard, (ii)
\PP-hard, and (iii) belong to the counting hierarchy. In particular,
whilst being decidable in \CH, both problems seem to be harder than
\textsc{PosSLP}: for the latter problem, no nontrivial lower bound is
known.

\begin{table*}[t]
  \begin{center}
    \renewcommand{\arraystretch}{1}
    \begin{tabular}{|c|c||c|c|c|}
      \hline
      \centering Parikh vector encoding & size of $\Sigma$ & DFA & NFA & CFG\\
      \hline
      \centering \multirow{3}{*}{unary} & unary & in $\LOGSPACE$ (\ref{prop:dfa-unary-upper}) & \NL-complete~(\ref{prop:dfa-unary-upper}) & \P-complete~(\ref{prop:dfa-unary-upper}) \\
      \cline{2-5}
       & fixed & \PL-complete (\ref{prop:dfa-results}) & 
      \multicolumn{2}{c|}{} \\
      \cline{2-3}
      & variable & \multicolumn{3}{c|}{$\PP$-complete (\ref{prop:dfa-results}, \ref{prop:nfa-cfg-variable-upper}, \ref{prop:nfa-cfg-fixed-lower})}  \\
      \hline
      \hline
      \centering \multirow{3}{*}{binary} & unary & in $\LOGSPACE$ (\ref{prop:dfa-unary-upper}) & \NL-complete (\ref{prop:dfa-unary-upper}) & \DP-complete~(\ref{prop:dfa-unary-upper})\\
      \cline{2-5}
       & fixed & PosMatPow-hard, in $\CH$ 
      (\ref{prop:dfa-results}, \ref{thm-parikh-binary-dfa}) & 
      \multirow{2}{*}{\PSPACE-complete (\ref{prop:nfa-cfg-variable-upper},\ref{prop:nfa-cfg-fixed-lower})} & 
      \multirow{2}{*}{$\PEXP$-complete (\ref{prop:nfa-cfg-variable-upper}, \ref{prop:nfa-cfg-fixed-lower})}\\
      \cline{2-3} 
      & variable & \textsc{PosSLP}-hard~\cite{HK15}, in 
      $\CH$ (\ref{thm-parikh-binary-dfa}) & & \\
      \hline
    \end{tabular}
    \label{tab:main-results}
    \caption{The complexity landscape of \PIC. References to
      propositions proving the stated complexity bounds are in parentheses.}
  \end{center}
\end{table*}

As discussed above, the cost problem is closely related to \BitP\ and
\PIC\ (for DFA and Parikh vectors encoded in binary). In fact, the
\PosSLP\ lower bound for the cost problem is strongly based on the
\PosSLP-hardness of \PIC~\cite[Prop.~5]{HK15}. This tight relationship
between the two classes of problems is our main motivation for
studying in the second part of this paper the complexity of \BitP\ and
\PIC\ also for other variants: Instead of a DFA, one can specify the
language by an NFA or even a context-free grammar (CFG).
Indeed, Kopczy{\'n}ski~\cite{Kop15} recently asked about the complexity of computing the number of words with a given Parikh image accepted by a CFG. Our results on the complexity of \PIC\ are collected in Table~\ref{tab:main-results} and similar results also hold for \BitP.
As the table shows, the complexity may depend on the encoding of the numbers in the Parikh vector
(unary or binary) and the size of the alphabet (unary alphabet, fixed
alphabet or a variable alphabet that is part of the input).

Perhaps remarkably, we show that \PIC\ for DFA over a two-letter alphabet and Parikh vectors
encoded in binary is hard for \PosMatPow. The latter problem was recently introduced by
Galby, Ouaknine and Worrell~\cite{GOW15-MatrixPowering} and asks,
given a square integer matrix $M \in \Z^{m \times m}$, a linear
function $f\colon \Z^{m \times m} \to \Z$ with integer coefficients, and a
positive integer~$n$, whether $f(M^n) \ge 0$, where all numbers in
$M$, $f$ and~$n$ are encoded in binary. Note that the entries of $M^n$
are generally of size exponential in the size of~$n$.  It is shown
in~\cite{GOW15-MatrixPowering} that \PosMatPow\ can be decided in
polynomial time for fixed dimension $m=2$.  The same holds for $m=3$
provided that $M$ is given in unary~\cite{GOW15-MatrixPowering}.  The
general \PosMatPow\ problem is in~\CH; in fact, it is is reducible to
\PosSLP, but the complexity of \PosMatPow\ is left open
in~\cite{GOW15-MatrixPowering}. In particular, it is not known whether
\PosMatPow\ is easier to decide than \PosSLP. Our result that \PIC\ is
\PosMatPow-hard already for a fixed-size alphabet while
\PosSLP-hardness seems to require an alphabet of variable size~\cite{HK15} could
be seen as an indication that \PosMatPow\ is easier to decide than \PosSLP.

\subsection{Related Work}\label{ssec:related-work}

The problems studied in this paper lie at the intersection of
probabilistic verification, automata theory, enumerative combinatorics
and computational complexity. Hence, there is a large body of related
work that we briefly discuss here.

\smallskip
\noindent
\textit{Probabilistic Verification.} Over the last decade and in
particular in recent years, there has been strong interest in
extensions of Markov chains and Markov decision processes (MDPs) with
(multi-dimensional) weights. Laroussinie and Sproston were the first
to show that model checking PCTL on cost chains, a generalisation of
the cost problem, is \NP-hard and  in \EXPTIME~\cite{LS05}; the
lower bound has recently been improved to \PP\ in~\cite{HK16}.
Qualitative aspects of quantile problems in weighted MDPs where the
probability threshold $\tau$ is either 0 or 1 have been studied by
Ummels and Baier in~\cite{UB13}, and iterative linear programming
based approaches for solving reward quantiles have been described
in~\cite{BDDKK14}. There is also a large body of work on synthesising
strategies for weighted MDPs that ensure both worst case as well as
expected value guarantees~\cite{BFRR14,RRS15a,CR15,CKK15};
see~\cite{RRS15b} for a survey on such beyond-worst-case-analysis
problems. In addition, some other \textsc{PosSLP}-hard numerical
problems have recently been discovered, for instance the problem of
finding mixed strategy profiles close to exact Nash equilibria in
three-person games \cite{EtessamiY10} or the model-checking problem of
interval Markov chains against unambiguous B\"uchi automata
\cite{BenediktLW13}.  These problems also belong to the counting
hierarchy -- it would be interesting to investigate whether they are
also \PP-hard.

\smallskip
\noindent
\textit{Counting Paths.}  We use the BEST theorem in the proofs of our
$\CH$ upper bounds for \BitP\ and \PIC\ (for DFA with Parikh vectors
encoded in binary) in order to count the number of Eulerian circuits
in a directed multi-graph. This multi-graph is succinctly given: There
may be exponentially many edges between two nodes. For an explicitly
given directed graph, the BEST theorem allows to compute the number of
Eulerian circuits in $\NC^2$ since it basically reduces the
computation to a determinant. On the other hand, it is well known that
computing the number of Eulerian circuits in an undirected graph is
\sharpP-complete \cite{BrightwellW05}. For directed (resp.,
undirected) graphs of bounded tree width, the number of Eulerian
circuits can be computed in logspace \cite{BalajiD15} (resp. $\NC^2$
\cite{BalajiDG15}). Finally, Mahajan and Vinay show~\cite{MV97} that
the determinant of a matrix has a combinatorial interpretation in
terms of the difference of the number of paths between two pairs of
nodes in a directed acyclic graph. In an analogous way, \PIC\ could be
viewed as a combinatorial interpretation of \PosSLP.

\smallskip
\noindent
\textit{Counting Words.} A problem related to the problem \PIC\ is the
computation of the number of all words of a given length $n$ in a
language $L$. If $n$ is given in unary encoding, then this problem can
be solved in $\NC^2$ for every fixed unambiguous context-free language
$L$ \cite{DBLP:journals/tcs/BertoniGS91}. On the other hand, there
exists a fixed context-free language $L \subseteq \Sigma^*$ (of
ambiguity degree two) such that if the function $a^n \mapsto \#(L \cap
\Sigma^n)$ can be computed in polynomial time, then $\EXPTIME =
\NEXPTIME$ \cite{DBLP:journals/tcs/BertoniGS91}.  Counting the number
of words of a given length encoded in unary that are accepted by a
given NFA (which is part of the input in contrast to the results of
\cite{DBLP:journals/tcs/BertoniGS91}) is \sharpP-complete
\cite[Remark~3.4]{KusL08}. The corresponding problem for DFA is
equivalent to counting the number of paths between two nodes in a
directed acyclic graph, which is the canonical \sharpL-complete
problem. Note that for a fixed alphabet and Parikh vectors encoded in
unary, the computation of $N(\A,\vec{p})$ for an NFA (resp.~DFA) $\A$
can be reduced to the computation of the number of words of a given
length encoded in unary accepted by an NFA (resp.~DFA) $\A'$: In that
case, one can easily compute in logspace a DFA $\A_{\vec{p}}$ for
$\parikh^{-1}(\vec{p})$ and then construct the product automaton of
$\A$ and $\A_{\vec{p}}$.

\section{Preliminaries}
\makeatletter{}
\subsection{Counting Problems for Parikh Images}

Let $\Sigma=\{a_1,\ldots,a_m\}$ be a finite alphabet.
A Parikh vector is vector of $m$ non-negative integers, i.e., an element of~$\N^m$.
Let $u\in
\Sigma^*$ be a word. For $a\in \Sigma$, we denote by $\abs{u}_a$ the
number of times $a$ occurs in $u$. The Parikh image $\parikh(u)\in
\N^m$ of $u$ is the Parikh vector counting how often
every alphabet symbol of $\Sigma$ occurs in $u$, i.e.,
$\parikh(u)\defeq (\abs{u}_{a_1},\ldots,\abs{u}_{a_m})$. The Parikh
image of a language $L\subseteq \Sigma^*$ is defined as
$\parikh(L)\defeq \{ \parikh(u) : u \in L \}\subseteq \N^m$.

We use standard language accepting devices in this paper.
A non-deterministic finite-state automaton (NFA) is a tuple
$\A=(Q,\Sigma,q_0,F,\Delta)$, where $Q$ is a finite set of control
states, $\Sigma$ is a finite alphabet, $q_0\in Q$ is an initial state,
$F\subseteq Q$ is a finite set of final states, and $\Delta\subseteq
Q\times \Sigma\times Q$ is a finite set of transitions. We write $p
\xrightarrow{a} q$ whenever $(p,a,q)\in \Delta$ and define
$\Delta(p,a)\defeq \{ q\in Q : p \xrightarrow{a} q\}$. For
convenience, we sometimes label transitions with words $w\in
\Sigma^+$. Such a transition corresponds to a chain of transitions that are
consecutively labelled with the symbols of $w$. We call $\A$ a
deterministic finite-state automaton (DFA) if $\# \Delta(q,a) \le 1$
for all $q\in Q$ and $a\in \Sigma$. Given $u = a_1 a_2\cdots a_n \in
\Sigma^*$, a run $\run$ of $\A$ on $u$ is a finite sequence of control
states $\run= p_0p_1\cdots p_n$ such that $p_0 = q_0$ and $p_i
\xrightarrow{a_i} p_{i+1}$ for all $0 \le i<n$. We call $\run$
accepting whenever $p_n\in F$ and define the language accepted by $\A$
as $L(\A)\defeq \{ u\in \Sigma^* : \A $ has an accepting run on $u\}$.
Finally, context-free grammars (CFG) are defined as usual.

Let $\Sigma$ be an alphabet of size $m$ and $\vec{p}\in \N^m$ be a
Parikh vector. For a language acceptor $\A$, we denote by
$N(\A,\vec{p})$ the number of words in $L(\A)$ with Parikh image
$\vec{p}$, i.e.,
\[
N(\A,\vec{p}) \defeq \#\{ u\in L(\A) : \parikh(u) = \vec{p} \}.
\]
We denote the counting function that maps $(\A,\vec{p})$ to
$N(\A,\vec{p})$ also with \countP. For complexity considerations, we
have to specify (i) the type of $\A$ (DFA, NFA, CFG), (ii) the
encoding of (the numbers in) $\vec{p}$ (unary or binary), and (iii)
whether the underlying alphabet is fixed or part of the input
(variable). For instance, we speak of \countP\ for DFA over a fixed
alphabet and Parikh vectors encoded in binary.  The same terminology
is used for the following computational problems:

\medskip
\problemx
{\PIC}
{Language acceptors $\A,\B$ over an alphabet $\Sigma$ of size $m$ and a
Parikh vector $\vec{p}\in \N^m$.}
{Is $N(\A,\vec{p}) > N(\B,\vec{p})$?}
\medskip

\problemx
{\BitP}
{Language acceptor $\A$ over an alphabet $\Sigma$ of size $m$, a
Parikh vector $\vec{p}\in \N^m$, and a number $i \in \N$ encoded binary .}
{Is the $i$-th bit of $N(\A,\vec{p})$ equal to one?}
\medskip

\noindent
Note that for a Parikh vector $\vec{p}$ encoded in binary, the number
$N(\A,\vec{p})$ may be doubly exponential in the input length (size of
$\A$ plus number of bits in $\vec{p}$). Hence, the number of bits in
$N(\A,\vec{p})$ can be exponential, and a certain position in the
binary encoding of $N(\A,\vec{p})$ can be specified with polynomially
many bits.

Our main result for the problems above is (see
Section~\ref{ssec:complexity} below for the formal definition of the
counting hierarchy):

\begin{theorem} \label{thm-parikh-binary-dfa}
For DFA over a variable alphabet and Parikh vectors encoded in binary,
the problems \BitP\ and \PIC\ belong to the counting hierarchy.
\end{theorem}
We postpone the proof of Theorem~\ref{thm-parikh-binary-dfa}  to Section~\ref{ssec-best}.
For other settings, the complexity of  \BitP\ and \PIC\  (and also the counting problem \countP)
will be studied in Section~\ref{sec-further}.

\subsection{Graphs} \label{sec-graphs}

A (finite directed) multi-graph is a tuple $G = (V,E,s,t)$, where $V$
is a finite set of nodes, $E$ is a finite set of edges, and the mapping
$s\colon E \to V$ (resp., $t\colon E \to V$) assigns to each edge its source
node (resp., target node). A loop is an edge $e \in E$ with
$s(e)=t(e)$.  A path (of length $n$) in $G$ from $u$ to $v$ is a sequence of edges $e_1, e_2, \ldots,
e_n$ such that $s(e_1) = u$, $t(e_n) = v$, and $t(e_i) = s(e_{i+1})$ for all $1 \leq i \leq n-1$.  We
say that $G$ is connected if for all nodes $u,v \in V$ there exists a
path in $G$ from $u$ to $v$.
We say that $G$ is loop-free if $G$ does not have loops.
The in-degree of $v$ is $d_G^+(v) \defeq \# t^{-1}(v)$ (note that the
preimage $t^{-1}(v)$ is the set of all incoming edges for node $v$) and the out-degree of $v$ is
$d_G^-(v) \defeq \# s^{-1}(v)$.

An edge-weighted multi-graph is a tuple $G = (V,E,s,t,w)$, where
$(V,E,s,t)$ is a multi-graph and $w \colon E \to \mathbb{N}$ assigns a
weight to every edge. We can define the ordinary multi-graph
$\tilde{G}$ induced by $G$ by replacing every edge $e \in E$ by $k =
w(e)$ many edges $e_1, \ldots, e_k$ with $s(e_i) = s(e)$ and $t(e_i) =
t(e)$. For $u, v \in V$ and $n \in \N$, define $N(G,u,v,n)$ as
the number of paths in~$\tilde{G}$ from $u$ to~$v$ of length~$n$.
Moreover, we set $d_G^-(v) = d_{\tilde{G}}^-(v)$ and $d_G^+(v) = d_{\tilde{G}}^+(v)$.

\subsection{Computational Complexity}\label{ssec:complexity}

We assume familiarity with basic complexity classes such as
\LOGSPACE\ (deterministic logspace), \NL, \P, \NP, \PH\ (the polynomial time hierarchy)
and \PSPACE. The class \DP\ is the class of all intersections $K \cap L$ with $K \in \NP$ and $L \in \coNP$.
Hardness for a complexity
class will always refer to logspace reductions.

A counting problem is a function $f \colon \Sigma^* \to \N$ for a
finite alphabet $\Sigma$. A counting class is a set of counting
problems. A logspace reduction from a counting problem $f \colon
\Sigma^* \to \N$ to a counting problem $g \colon \Gamma^* \to \N$ is a
logspace computable function $h\colon \Sigma^* \to \Gamma^*$ such that
for all $x \in \Sigma^*$: $f(x) = g(h(x))$. Note that no
post-computation is allowed. Such reductions are also called
parsimonious. Hardness for a counting class will always refer to
parsimonious logspace reductions.

The counting class \sharpP\ contains all functions $f\colon \Sigma^*
\to \mathbb{N}$ for which there exists a non-deterministic
polynomial-time Turing machine $M$ such that for every $x \in
\Sigma^*$, $f(x)$ is the number of accepting computation paths of $M$
on input $x$. The class \PP\ (probabilistic polynomial time) contains
all problems $A$ for which there exists a non-deterministic
polynomial-time Turing machine $M$ such that for every input $x$, $x
\in A$ if and only if more than half of all computation paths of $M$
on input $x$ are accepting. By a famous result of Toda \cite{To91},
$\PH \subseteq \P^{\PP}$, where $\P^{\PP}$ is the class of all
languages that can be decided in deterministic polynomial time with
the help of an oracle from \PP.  Hence, if a problem is \PP-hard, then
this can be seen as a strong indication that the problem does not
belong to \PH\ (otherwise \PH\ would collapse). If we replace in the
definition of \sharpP\ and \PP\ non-deterministic polynomial-time
Turing machines by non-deterministic logspace Turing machines (resp.,
non-deterministic polynomial-space Turing machines; non-deterministic
exponential-time Turing machines), we obtain the classes \sharpL\ and
\PL\ (resp., \sharpPSPACE\ and \ComplexityFont{PPSPACE};
\sharpEXP\ and \PEXP).  Ladner \cite{Lad89} has shown that a function
$f$ belongs to
\sharpPSPACE\ if and only if for a given input $x$ and a binary
encoded number $i$ the $i$-th bit of $f(x)$ can be computed in
\PSPACE. It follows that \ComplexityFont{PPSPACE} = \PSPACE.  It is
well known that \PP\ can be also defined as the class of all
languages $L$ for which there exist two \sharpP-functions $f_1$ and
$f_2$ such that $x \in L$ if and only if $f_1(x) > f_2(x)$ and
similarly for \PL\ and \PEXP.

The levels of the {\em counting hierarchy} $\C^p_i$ ($i \geq 0$) are
inductively defined as follows: $\C^p_0 = \P$ and $\C^p_{i+1} =
\PP^{\C^p_{i}}$ (the set of languages accepted by a $\PP$-machine as
above with an oracle from $\C^p_{i}$) for all $i \geq 0$. Let $\CH =
\bigcup_{i\geq 0} \C^p_{i}$ be the counting hierarchy. It is not
difficult to show that $\CH \subseteq \PSPACE$, and most complexity
theorists conjecture that $\CH \subsetneq \PSPACE$. Hence, if a
problem belongs to the counting hierarchy, then the problem is probably not \PSPACE-complete.
  
The circuit complexity class \DLOGTIME-uniform
$\TC^0$ is the class of all languages that can be decided with a
constant-depth polynomial-size \DLOGTIME-uniform circuit family of
unbounded fan-in that in addition to normal Boolean gates (AND, OR and
NOT) may also use threshold gates. \DLOGTIME-uniformity means that one
can compute in time $O(\log n)$ (i) the type of a given gate of the
$n$-th circuit, and (ii) whether two given gates of the $n$-th circuit
are connected by a wire. Here, gates of the $n$-th circuit are encoded
by bit strings of length $O(\log n)$.  If we do not allow threshold
gates in this definition, we obtain \DLOGTIME-uniform $\AC^0$.  The
class $\NC^k$ ($k \geq 1$) is the class of all of all languages that
can be decided with a polynomial-size \DLOGTIME-uniform circuit family
of depth $O(\log^k n)$, where only Boolean gates of fan-in two are
allowed. \DLOGTIME-uniform $\TC^0$ is contained in \DLOGTIME-uniform
$\NC^1$.
  
There are obvious generalization of the above language classes
$\AC^0$, $\TC^0$, and $\NC^k$ to function classes.  Given two $n$-bit
numbers $x,y\in \Z$, $x+y$ (resp., $x \cdot y$)
can be computed in \DLOGTIME-uniform $\AC^0$ (resp., \DLOGTIME-uniform
$\TC^0$) \cite{Vol99}.  Even the product of $n$ numbers $x_1,\ldots,x_n
\in \N$, each of bit-size at most $n$, and matrix powers $A^n$ with
$n$ given in unary and $A$ of constant dimension can be computed in
\DLOGTIME-uniform
$\TC^0$~\cite{HAMB02,DBLP:conf/mfcs/AllenderBD14}. If $n$ is given in
binary (and $A$ has again constant dimension) then the computation of
a certain bit of $A^n$ can be done in
$\PH^{\PP^{\PP^\PP}}$~\cite{DBLP:conf/mfcs/AllenderBD14}.  Finally,
computing the determinant of an integer matrix with entries encoded in
binary is in $\NC^2$.  More details on the counting hierarchy (resp.,
circuit complexity) can be found in \cite{AllenderW90} (resp.,
\cite{Vol99}).

\section{A Toolbox for the Counting Hierarchy}\label{sec-aux}
\makeatletter{}In the subsequent sections, in order to show our \CH\ upper bounds we
require some closure results for the counting hierarchy. These
results are based on ideas and results developed
in~\cite{ABKB09,Burgisser09}, which are, however, not sufficiently
general for our purposes.

Let $\mathbb{B} \defeq \{0,1\}^*$ in the following definitions. For a
$k$-tuple $\overline{x} = (x_1, \ldots, x_k) \in \mathbb{B}^k$ let
$|\overline{x}| = \sum_{i=1}^k |x_i|$. The following definition 
is a slight variant of the definition in \cite{Burgisser09} that
suits our purposes better. Consider a function $f\colon \mathbb{B}^k
\to \mathbb{N}$ that maps a $k$-tuple of binary words to a natural
number. We say that $f$ is in $\CH$ if there exists a polynomial
$p(n)$ such that the following holds:
\begin{itemize}
\item For all $\overline{x} \in \mathbb{B}^k$ we have $f(\overline{x})
  \leq 2^{2^{p(|\overline{x}|)}}$.
\item The set of tuples
$$L_f \defeq \{ (\overline{x}, i) \in \mathbb{B}^k \times \mathbb{N} :
  \text{the $i$-th bit of $f(\overline{x})$ is equal to $1$}\}$$
  belongs to $\CH$ (here, we assume that $i$ is given in binary
  representation).
\end{itemize}
We also consider mappings $f \colon \prod_{j=1}^k D_j \to \mathbb{N}$,
where the domains $D_1, \ldots, D_k$ are not $\mathbb{B}$. In this
case, we assume some standard encoding of the elements from $D_1,
\ldots, D_k$ as words over a binary alphabet.

\begin{lemma} \label{lemma-add-mult-CH}
  If the function $f\colon \mathbb{B}^{k+1} \to \mathbb{N}$ belongs to
  $\CH$ and $p(n)$ is a polynomial, then also the functions $g\colon
  \mathbb{B}^k \to \mathbb{N}$ and $h\colon \mathbb{B}^k \to
  \mathbb{N}$ belong to $\CH$, where
\begin{align*}
g(\overline{x}) & =\!\! \sum_{y \in \{0,1\}^{p(|\overline{x}|)}}  f(\overline{x},y)
   & \text{and} & &
   h(\overline{x}) & =\!\!  \prod_{y \in \{0,1\}^{p(|\overline{x}|)}} \!\!  f(\overline{x},y) .
\end{align*}
\end{lemma}

\begin{proof}
  We only prove the statement for products, the proof for sums is the
  same. To this end, we follow the arguments from~\cite{ABKB09}
  showing that \textsc{BitSLP} belongs to the counting hierarchy.
  As mentioned in Section~\ref{ssec:complexity},
  iterated product, i.e., the problem of computing the product of a
  sequence of binary encoded integers, belongs to $\DLOGTIME$-uniform
  $\TC^0$.  More precisely, there is a $\DLOGTIME$-uniform $\TC^0$ circuit family,
  where the $n$-th circuit $C_n$ has $n^2$ many input gates, which are
  interpreted as $n$-bit integers $x_1, \ldots, x_n$, and $n^2$ output
  gates, which evaluate to the bits in the product $\prod_{i=1}^n
  x_i$.  Let $c$ be the depth of the circuits $C_n$, which is a fixed
  constant.

  Since the function $f$ belongs to \CH, there is a polynomial $q(n)$
  such that for all $\overline{x} \in \mathbb{B}^k$ and $y \in
  \mathbb{B}$ we have $f(\overline{x},y) \leq
  2^{2^{q(|\overline{x}|+|y|)}}$. Let $r(n)$ be the polynomial with
  $r(n) = q(n + p(n))$. Hence, for all $\overline{x} \in
  \mathbb{B}^{k}$ and $y \in \{0,1\}^{p(|\overline{x}|)}$ we have
  $f(\overline{x},y) \leq 2^{2^{r(|\overline{x}|)}}$. We can assume
  that $r(n) \geq p(n)$ for all $n$ (simply assume that $q(n) \geq
  n$).

For an input tuple $\overline{x} \in \mathbb{B}^k$ with $|\overline{x}|=n$ we consider the circuit $D_n \defeq C_{2^{r(n)}}$.
It takes $2^{r(n)} \geq 2^{p(n)}$ integers with $2^{r(|\overline{x})}$ bits as input.
Hence, we can consider the input tuple 
$$
\overline{z} = (f(\overline{x},y_1), f(\overline{x},y_2), \ldots, f(\overline{x},y_{2^{p(n)}}), 1, \ldots, 1)
$$ for $D_n$.  Here, $y_1, \ldots, y_{2^{p(n)}}$ is the lexicographic
enumeration of all binary words of length $p(n)$. We pad the tuple
with a sufficient number of ones so that the total length of the tuple
is $2^{r(n)}$.

There is a polynomial $s(n)$ such that $D_n$ has at most $2^{s(n)}$ many gates. Hence,
a gate of $D_n$ can be identified with a bitstring of length $s(n)$. Then, one shows that
for every level $1 \leq i \leq c$ (where level 1 consists of the input gates) the following set belongs to $\CH$:
\begin{multline*}
  \{ (\overline{x}, u) : \overline{x} \in \mathbb{B}^k , u \in
  \{0,1\}^{s(|\overline{x}|)}, \text{ gate $u$ belongs
    to}\\ \text{level $i$ of $D_{|\overline{x}|}$ and evaluates to $1$
      if $\overline{z}$ is the input for $D_{|\overline{x}|}$}\}.
\end{multline*}
This is shown by a straightforward induction on $i$ as in
\cite{ABKB09}.  For the induction base $i = 1$ one uses the fact that
the function $f$ belongs to $\CH$.

For the statement about sums, the proof is the same using the result that iterated
sum belongs to $\DLOGTIME$-uniform $\TC^0$ as well (which is much easier to show than
the corresponding result for iterated products).
\end{proof}

\begin{remark} \label{remark-factorial}
  A particular application of Lemma~\ref{lemma-add-mult-CH} that we
  will use in Section~\ref{sec-mc} and Section~\ref{ssec-best} is
  the following: Assume that $f, g \colon \mathbb{B}^k \to \N$ are
  functions such that for a given tuple $\overline{x} \in
  \mathbb{B}^k$ the values $f(\overline{x})$ and $g(\overline{x})$ are
  bounded by $2^{\text{poly}(|\overline{x}|)}$ and the binary
  representations of these numbers can be computed in polynomial
  time. Then, the mappings defined by $f(\overline{x})!$ and
  $h(\overline{x}) = f(\overline{x})^{g(\overline{x})}$ belong to \CH.
\end{remark}

\begin{remark} \label{remark-restricted-sum-product}
  In our subsequent applications of Lemma~\ref{lemma-add-mult-CH} we
  have to consider the case that $g$ is given as
    \[
  g(\overline{x}) = \sum_{y \in S(\overline{x}) \cap \{0,1\}^{p(|\overline{x}|)}} \!\!\! f(\overline{x},y),
  \]
    such that for a given tuple
  $\overline{x} \in \mathbb{B}^k$ and a binary word $y \in
  \{0,1\}^{p(|\overline{x}|)}$ one can decide in polynomial time
  whether $y \in S(\overline{x})$.  This case can be easily reduced to
  Lemma~\ref{lemma-add-mult-CH}, since $g(\overline{x}) = \sum_{y \in
    \{0,1\}^{p(|\overline{x}|)}} f'(\overline{x},y)$, where $f'$ is
  defined as
  \[
  f'(\overline{x},y) \defeq \begin{cases}
    f(\overline{x},y) & \text{ if } y \in  S(\overline{x}) \\
    0  & \text{ otherwise. }
  \end{cases} 
  \]
Moreover, if $f$ belongs to \CH, then also $f'$ belongs to \CH.
The same remark applies to products instead of sums.
\end{remark}

\begin{lemma} \label{lemma-division-CH}
  If the functions $f\colon \mathbb{B}^k \to \mathbb{N}$ and $g\colon
  \mathbb{B}^k \to \mathbb{N}$ belong to $\CH$, then also the following functions
  $q\colon \mathbb{B}^k \to \mathbb{N}$ (quotient) and $d\colon \mathbb{B}^k
  \to \mathbb{N}$ (modified difference) belong to \CH:
$$
q(\overline{x}) \defeq  \bigg\lfloor \frac{f(\overline{x})}{g(\overline{x})}  \bigg\rfloor \quad \text{ and } \quad
d(\overline{x}) \defeq \max \{0, f(\overline{x})-g(\overline{x})\}
$$
\end{lemma}

\begin{proof}
  The proof is the same as for Lemma~\ref{lemma-add-mult-CH}, using
  the result that division (resp., subtraction) of integers encoded in
  binary is in $\DLOGTIME$-uniform $\TC^0$ \cite{HAMB02} (resp.,
  $\AC^0 \subseteq \TC^0$ \cite{Vol99}).
\end{proof}
In particular, we have:

\begin{lemma} \label{lemma-comparision-CH}
  If the functions $f\colon \mathbb{B}^k \to \mathbb{N}$ and $g\colon
  \mathbb{B}^k \to \mathbb{N}$ belong to $\CH$, then also the following function
  $h\colon \mathbb{B}^k \to \mathbb{N}$ belongs to \CH:
$$
h(\overline{x}) \defeq \begin{cases} 
 1 & \text{ if } f(\overline{x}) > g(\overline{x}) \\
 0 & \text{ otherwise }
\end{cases}
$$
\end{lemma}

\section{Cost Problems in Markov Chains}\label{sec-mc}
\makeatletter{}
A \emph{Markov chain} is a triple $\M = (S,s_0,\delta)$, where $S$ is
a countable (finite or infinite) set of states, $s_0 \in S$ is an
initial state, and $\delta\colon S \to \dist(S)$ is a probabilistic
transition function that maps a state to a probability distribution
over the successor states.  Given a Markov chain we also write $s
\tran{p} t$ or $s \tran{} t$ to indicate that $p = \delta(s)(t) > 0$.
A \emph{run} is an infinite sequence $s_0 s_1 \cdots \in \{s_0\}
S^\omega$ with $s_i \tran{} s_{i+1}$ for $i \in \N$.  We write
$\Run(s_0 \cdots s_k)$ for the set of runs that start with $s_0 \cdots
s_k$.  We associate to~$\M$ the standard probability space
$(\Run(s_0),\F,\PM)$ where $\F$ is the $\sigma$-field generated by all
basic cylinders $\Run(s_0 \cdots s_k)$ with $s_0 \cdots s_k \in
\{s_0\} S^*$, and $\PM\colon \F \to [0,1]$ is the unique probability
measure such that $\PM(\Run(s_0 \cdots s_k)) = \prod_{i=1}^{k}
\delta(s_{i-1})(s_i)$.

A \emph{cost chain} is a tuple $\chain = (Q, q_0, t, \Delta)$, where
$Q$ is a finite set of control states, $q_0 \in Q$ is the initial
control state, $t \in Q$ is the target control state, and
$\Delta\colon Q \to \dist(Q \times \N)$ is a probabilistic transition
function.  Here, for $q, q' \in Q$ and $k \in \N$, when $q$ is the
current control state, the value $\Delta(q) (q',k) \in [0,1]$ is the
probability that the cost chain transitions to control state~$q'$ and
cost~$k$ is incurred.  For the complexity results we define the
\emph{size} of~$\chain$ as the size of a succinct description, i.e.,
the costs are encoded in binary, the probabilities are encoded as
fractions of integers in binary (so the probabilities are rational),
and for each $q \in Q$, the distribution $\Delta(q)$ is described by
the list of triples $(q',k,p)$ with $\Delta(q)(q',k) = p > 0$ (so we
assume this list to be finite). We define the set $E$ of edges of
$\chain$ as $E\defeq \{ (q,k,q') : \Delta(q)(q',k) > 0\}$, and write
$\Delta(e)$ and $k(e)$ for $\Delta(q)(q',k)$ and $k$ respectively, whenever $e=(q,k,q')\in E$. A cost chain~$\chain$
induces a Markov chain $\M_\chain = (Q \times \N, (q_0,0), \delta)$
with $\delta(q,c)(q',c') = \Delta(q)(q',c'-c) $ for all $q,q' \in Q$
and $c, c' \in \N$ with $c' \geq c$. For a state $(q,c) \in Q \times
\N$ in~$\M_\chain$ we view~$q$ as the current control state and $c$ as
the current cost, i.e., the cost accumulated so far.

In this section, we will be interested in the cost accumulated during
a run before reaching the target state~$t$. Following~\cite{HK15},
we assume (i) that the target state~$t$ is almost surely
reached, and (ii) that $\Delta(t)(t,0)=1$, hence runs that visit $t$
do not leave $t$ and accumulate only a finite cost. Those assumptions
are needed for the following definition to be
sound\footnote{See~\cite{HK15} for a discussion on why
  those assumptions can be made.}: Given a cost chain~$\chain$, we
define a random variable $K_\chain \colon \Run((q_0,0)) \to \N$ such
that $K_\chain((q_0,0) \ (q_1,c_1) \ \cdots) = c$ if there exists $i
\in \N$ with $(q_i,c_i) = (t,c)$.
We view $K_\chain(w)$ as the accumulated cost of a run~$w$. From the
aforementioned assumptions on~$t$, it follows that the random
variable~$K_\chain$ is almost surely defined. Furthermore, we can
assume without loss of generality that all zero-cost edges lead
to~$t$:
\begin{lemma} \label{lemma-zero-cost-edges}
  Given a cost chain~$\chain$, one compute in polynomial time a cost
  chain~$\chain'$ such that the distributions of $K_\chain$ and
  $K_{\chain'}$ are equal and all zero-cost edges in $\chain'$ lead to
  the target~$t$.
\end{lemma}
\begin{proof}
  The idea is to contract paths that consist of a sequence of
  zero-cost edges and a single edge that is either labelled with a
  non-zero cost or leads to~$t$.  In more detail, one proceeds as
  follows.  Let $E'$ be the set of edges that have non-zero cost or
  lead to~$t$.  For all $e \in E'$ and $q \in Q$, define $x_{q,e}^*$
  as the probability that, in a random sequence of edges starting
  from~$q$, the edge~$e$ is the first edge in~$E'$.  Let $e \in E'$.
  Using graph reachability, it is easy to compute $Q_e := \{q \in Q :
  x_{q,e}^* > 0\}$. It is also straightforward to set up a system of
  linear equations with a variable~$x_{q,e}$ for each $q \in Q_e$ such
  that the $x_{q,e}$-component of the (unique) solution is equal to
  $x_{q,e}^*$.  Hence, the numbers~$x_{q,e}^*$ can be computed in
  polynomial time.  Finally, construct~$\chain'$ from~$\chain$ as
  follows: Remove all edges from $E \setminus E'$, and for all $q, e$
  with $x_{q,e}^*>0$ add an edge from $q$ to the target of~$e$ such
  that the new edge has probability~$x_{q,e}^*$ and the same cost
  as~$e$ (or add $x_{q,e}^*$ to the probability if such an edge
  already exists).
\end{proof} 
Let $x$ be a fixed variable. An \emph{atomic cost formula} is an
inequality of the form $x \le b$ where $b \in \N$ is encoded in
binary, and a  \emph{cost formula} is an arbitrary Boolean combination of
atomic cost formulas. We say that a number $n \in \N$ \emph{satisfies}
a cost formula~$\varphi$, in symbols $n \models \varphi$, if $\varphi$
is true when $x$ is replaced by~$n$. Let $\links \varphi \rechts
\defeq \{ n : n \models \varphi \}$.

In \cite{HK15} the first two authors studied the following problem:
\smallskip
\problemx
{Cost Problem}
{A cost chain $\chain$, a cost formula $\varphi$, and a probability
threshold $\tau\in [0,1]$ given as a fraction of integers encoded in binary.}
{Does $\PM(K_\chain \models \varphi) \ge \tau$ hold?}
\smallskip
It was shown in~\cite{HK15} that the cost problem is $\PP$-hard,
\textsc{PosSLP}-hard and in \PSPACE.  Motivated by the problem
\BitSLP, we also consider the following related problem which allows
us to extract bits of the probability that the cost satisfies a cost
formula:
\smallskip
\problemx
{BitCost}
{A cost chain $\chain$, a cost formula $\varphi$, and an integer $j \geq 0$ in binary encoding.}
{Is the $j$-th bit of $\PM(K_\chain \models \varphi)$ equal to one?}
\smallskip
By application of
Theorem~\ref{thm-parikh-binary-dfa}, in this section we improve  the
upper bound for the cost problem to \CH:

\begin{theorem}\label{thm:mc}
  The \textsc{Cost problem} and \BitProb\ belong to \CH.
\end{theorem}
\begin{proof}
Let us first show the statement for \BitProb.
Let $\chain$ be a cost chain and $\varphi$ be a cost formula. We first derive a formula for the probability
$\PM(K_\chain \models \varphi)$.
By Lemma~\ref{lemma-zero-cost-edges} we can assume that all zero-cost edges in $\chain'$ lead to
the target~$t$.
For any edge $e\in E$,
let $\Delta(e)=m_e/d_e$ be the probability of $e$. Recall that the numbers
$m_e$ and $d_e$ are given in binary notation.
Given a finite prefix of a run $u=s_0s_1\cdots s_k$ with the
corresponding set of edges $e_1,\ldots, e_k$, the Parikh image
$\vec{p}\colon E\to \N$ of $w$ is defined as expected and denoted by
$\parikh(u)$. Let $c$ be the maximum constant occurring in $\varphi$
and assume for now that the set $\links \varphi \rechts$
is finite; we will deal with the infinite case in due course. In
particular, this implies that $\links \varphi \rechts \subseteq \{0, \ldots, c\}$.
Given a Parikh vector
 $\vec{p}\colon E\to \N$, we define
\begin{multline*}
  W_{\vec{p}}\defeq \{ w\in \Run((q_0,0)) : w=u \cdot
  (t,k)^\omega,\\ u \in ((Q \setminus \{t\}) \times \N)^*, \parikh(u \cdot
  (t,k))=\vec{p} \}.
\end{multline*}
Note that $W_{\vec{p}}$ is a finite set, and denote its cardinality by
$N(\chain,\vec{p})\defeq \# W_{\vec{p}}$. Moreover, we set
\[
K(\vec{p})\defeq \sum_{e \in E} k(e) \cdot \vec{p}(e)
\] 
and, for $i \geq 0$, $K^{-1}(i) \defeq \{ \vec{p} \in \N^E :
K(\vec{p}) = i \}$.  With these definitions, we have:
\allowdisplaybreaks
\begin{align*}
\PM(K_\chain \models \varphi) & =  \sum_{i \in \links\varphi\rechts} \PM(K_\chain = i) \nonumber \\
& = \sum_{i \in \links\varphi\rechts} \sum_{\vec{p}\in K^{-1}(i)} \PM( W_{\vec{p}}) \nonumber \\
& = \sum_{i \in \links\varphi\rechts} \sum_{\vec{p}\in K^{-1}(i)} N(\chain,\vec{p})\cdot 
\prod_{e\in E} \Delta(e)^{\vec{p}(e)} \nonumber \\
& = \frac{\displaystyle \sum_{i \in \links\varphi\rechts} \sum_{\vec{p}\in K^{-1}(i)} N(\chain,\vec{p})\cdot 
\prod_{e\in E} m_e^{\vec{p}(e)}  d_e^{c+1-\vec{p}(e)}}{\displaystyle \prod_{e\in E} d_e^{c+1}}. \end{align*}
Note that $c+1-\vec{p}(e) \geq 0$ for every $i \in
\links\varphi\rechts$, $\vec{p}\in K^{-1}(i)$, and $e \in E$: indeed, $i \in
\links\varphi\rechts$ implies that $i \leq c$, and the fact that all
zero-cost edges lead to the target implies that $\sum_{e \in E}
\vec{p}(e) \leq c+1$ whenever $K(\vec{p}) = i \leq c$.

From the above formula for $\PM(K_\chain \models \varphi)$, it follows
that the $j$-th bit of $\PM(K_\chain \models \varphi)$ is the least significant
bit of the following integer:
\[
P(\chain,\varphi,j) \defeq 
\left\lfloor
\frac{\displaystyle \sum_{i \in \links\varphi\rechts} \sum_{\vec{p}\in K^{-1}(i)} \!\!\! 2^j \,  N(\chain,\vec{p})
\prod_{e\in E} m_e^{\vec{p}(e)}  d_e^{c+1-\vec{p}(e)}}{\displaystyle \prod_{e\in E} d_e^{c+1}} \right\rfloor
\]
Let us fix a standard binary encoding of the cost chain $\chain$ and
the cost formula $\varphi$ so that the definitions and results from
Section~\ref{sec-aux} are applicable.  By
Theorem~\ref{thm-parikh-binary-dfa}, the mapping $(\chain,\vec{p})
\mapsto N(\chain,\vec{p})$ can be computed in \CH. Hence,
Lemma~\ref{lemma-add-mult-CH} and \ref{lemma-division-CH} (see also
Remark~\ref{remark-factorial} and \ref{remark-restricted-sum-product}
after Lemma~\ref{lemma-add-mult-CH}) imply that the function
$P(\chain,\varphi,j)$ belongs to \CH\ (which implies that
\BitProb\ belongs to \CH). For this, note that for given $\chain$,
$\varphi$, $i \leq c$ and $\vec{p}$ with $\sum_{e \in E} \vec{p}(e)
\leq c+1$ one can decide in polynomial time whether $i \in
\links\varphi\rechts$ and $K(\vec{p}) = i$. This allows to apply
Remark~\ref{remark-restricted-sum-product} and concludes the proof in
case $\links \varphi \rechts$ is finite.

If $\links \varphi \rechts$ is not finite then $\links \neg \varphi
\rechts$ is finite. Moreover, we have $\PM(K_\chain \models \varphi) =
1 - \PM(K_\chain \models \neg\varphi)$.  Hence, we have to compute the
least significant bit of the number $\lfloor 2^j - 2^j \cdot \PM(K_\chain \models
\neg\varphi) \rfloor$. This can be done in \CH\ by using again the
above formula for $\PM(K_\chain \models \neg\varphi)$ and the lemmas
from Section~\ref{sec-aux}.

In order to show that the cost problem belongs to \CH, we can use the
same line of arguments. We first consider the case that $\links
\varphi \rechts$ is finite. Let $\tau = m/d$ be the threshold probability
from the cost problem. Then we have to check whether
$$
\sum_{i \in \links\varphi\rechts} \sum_{\vec{p}\in K^{-1}(i)} d \cdot N(\chain,\vec{p})\cdot 
\prod_{e\in E} m_e^{\vec{p}(e)}  d_e^{c+1-\vec{p}(e)} \geq m \cdot \prod_{e\in E} d_e^{c+1}.
$$
This can be checked in \CH\ as above, where in addition we have to use Lemma~\ref{lemma-comparision-CH}.
Finally, if $\links \varphi \rechts$ is not finite (but $\links \neg \varphi \rechts$ is finite), then we check as above whether
$\PM(K_\chain \models \neg \varphi) \le 1 - \tau$. 
\end{proof}

\section{CH Upper Bounds for DFA}\label{ssec-best}
\makeatletter{}In this section, we prove Theorem~\ref{thm-parikh-binary-dfa}.  To
this end, we will apply two classical results from graph theory:
Tutte's matrix-tree theorem and the BEST theorem which we introduce
first.
 
Let $G = (V,E,s,t)$ be a (finite directed) multi-graph as defined in
Section~\ref{sec-graphs}.  We call $G$ Eulerian if $d_G^-(v)=d_G^+(v)$ for
all $v\in V$. An Eulerian circuit is a path $e_1, e_2, \ldots, e_n$
such that $E = \{e_1, e_2, \ldots, e_n\}$, $e_i \neq e_j$ for $i \neq
j$, and $t(e_n) = s(e_1)$.  Let us denote by $e(G)$ the number of
Eulerian circuits of $G$, where we do not distinguish between the
Eulerian circuits $e_1, e_2, \ldots, e_n$ and $e_i, e_{i+1}, \ldots,
e_n, e_1, \ldots, e_{i-1}$. Alternatively, $e(G)$ is the number of
Eulerian circuits that start with a distinguished starting node.

Let $G= (V,E,s,t)$ be a connected loop-free multi-graph and assume
that $V = \{1,\ldots,n\}$. The adjacency matrix of $G$ is $A(G) =
(a_{i,j})_{1 \leq i,j \leq n}$, where $a_{i,j} \defeq \# \{ e \in E :
s(e)=i, t(e)=j\}$ is the number of edges from $i$ to $j$.  The
out-degree matrix $D^-(G) = (d_{i,j})_{1 \leq i,j \leq n}$ is defined
by $d_{i,j}\defeq 0$ for $i \neq j$ and $d_{i,i} \defeq d_G^-(i)$. The
Laplacian $L(G)$ of $G$ is $L(G) \defeq D^-(G) - A(G)$.  If $M$ is a
matrix then we denote by $M^{i,j}$ the $(i,j)$ minor of $M$, i.e., the
matrix obtained from $M$ by deleting its $i$-th row and $j$-th
column. By Tutte's matrix-tree theorem (see
e.g. \cite[p.~231]{Aigner2007}), the number $t(G,i)$ of directed
spanning trees oriented towards vertex $i$ (i.e., the number of
sub-graphs $T$ of $G$ such that (i) $T$ contains all nodes of $G$,
(ii) $i$ has out-degree $0$ in $T$ and (iii) for every other vertex $j
\in V \setminus \{i\}$ there is a unique path in $T$ from $j$ to $i$)
is equal to $(-1)^{i+j}\cdot \det(L(G)^{i,j})$, where $j \in V$ is
arbitrary. In particular, $t(G,i) = \det(L(G)^{i,i})$. Moreover, when
$G$ is Eulerian, then $t(G,i) = t(G,j)$ for all $i,j \in V$
\cite[p.~236]{Aigner2007} and we denote this number with $t(G)$.  If
$G$ is not loop-free then we set $t(G) \defeq t(G')$, where $G'$ is
obtained from $G$ by removing all loops. Since loops are not relevant
for counting the number of spanning trees, $t(G)$ still counts the
number of spanning trees (oriented towards an arbitrary vertex).

Assume now that $G= (V,E,s,t)$ is a connected Eulerian multi-graph and
assume without loss of generality that $V = \{1,\ldots,n\}$. Then, by
Tutte's matrix-tree theorem we have $t(G) = \det(L(G)^{1,1})$.  The
BEST theorem (named after de Bruijn, van Aardenne-Ehrenfest, Smith and
Tutte who discovered it) allows for computing the number $e(G)$ of
Eulerian paths of a connected Eulerian multi-graph, see
e.g. \cite[p.~445]{Aigner2007}:

\begin{theorem}[BEST theorem]\label{thm:best}
  Let $G= (V,E,s,t)$ be a connected Eulerian multi-graph. Then we have
  \[
  e(G) = t(G) \cdot \prod_{v \in V} (d_G^-(v) - 1)! .
  \]
\end{theorem}
Let $G= (V,E,s,t,w)$ be a connected edge-weighted multi-graph. 
Then $G$ is Eulerian if $\tilde{G}$ (the corresponding unweighted multi-graph)
is Eulerian, and in this case we define $e(G)$ as the number of 
paths $e_1,e_2,\ldots,e_n$ such that $t(e_n) = s(e_1)$ and for every edge
$e$, $w(e) = \# \{ i : 1 \leq i \leq n, e = e_i \}$.
The following result is then a straightforward corollary of the BEST theorem:
\begin{corollary}\label{coro:best}
  Let $G= (V,E,s,t,w)$ be a connected Eulerian edge-weighted multi-graph.  
  Then we have
  \[
  e(G) = t(\tilde{G}) \cdot \frac{\prod_{v \in V} (d_G^-(v) - 1)!}{\prod_{e \in E} w(e)!}
  \]
\end{corollary}

\begin{proof}
The result follows from the identity
$$
e(G) = \frac{e(\tilde{G})}{\prod_{e \in E} w(e)!}.
$$
To see this, note that every Eulerian circuit of $G$ corresponds to 
exactly $\prod_{e \in E} w(e)!$ many Eulerian circuits of $\tilde{G}$.
These Eulerian circuits are obtained by fixing for every $e \in E$ 
an arbitrary permutation of the $k = w(e)$ many edges $e_1, \ldots, e_k$ 
in $\tilde{G}$ that are derived from $e$. 
\end{proof}
We now turn towards the proof of
Theorem~\ref{thm-parikh-binary-dfa}. Let $\A=(Q,\Sigma,q_0,F,\Delta)$
be a DFA and fix a Parikh vector $\vec{p}$. We will use
Corollary~\ref{coro:best} to compute the number
$N(\A,\vec{p})$.  First define a new DFA $\A'$ as follows: Add a
fresh symbol $b$ to the alphabet together with all transitions
$(q_f, \omega, q_0)$ for $q_f \in F$; the initial state $q_0$ is the
only final state of $\A'$. Moreover, we extend the Parikh vector
$\vec{p}$ to a Parikh vector $\vec{p}'$ by $\vec{p}'(b)=1$ and
$\vec{p}'(a) = \vec{p}(a)$ for all $a \neq b$. Then we have
$N(\A,\vec{p}) = N(\A',\vec{p}')$. Hence, for the rest of this section
we may only consider DFA where the initial state is also the unique
final state. We call such a DFA {\em well-formed}. For a well-formed DFA
$\A=(Q,\Sigma,q_0,\{q_0\},\Delta)$ and $\vec{p}\colon \Sigma \to \N$
let $W(\A,\vec{p})$ be the set of all mappings $w\colon \Delta \to
\mathbb{N}$ such that the following two conditions hold:
\begin{itemize}
\item For every $a \in \Sigma$, we have $\vec{p}(a) = \sum_{(p,a,q)
  \in \Delta} w(p,a,q)$.
\item The edge-weighted multi-graph $\A^w \defeq (Q,\Delta,s,t,w)$
  (where $s(p,a,q) = p$ and $t(p,a,q)=q$) is connected and Eulerian.
\end{itemize}
With these definitions at hand, the following lemma is straightforward
to show.
\begin{lemma} \label{lemma-BEST->DFA}
Let $\A$ be a well-formed DFA. Then, we have
$$
N(\A,\vec{p}) = \sum_{w \in W(\A,\vec{p})} e(\A^w).
$$
\end{lemma}
We are now fully prepared to give the proof of
Theorem~\ref{thm-parikh-binary-dfa}. The statement for \PIC\ is an
immediate corollary of the statement for \BitP\ and
Lemma~\ref{lemma-comparision-CH}. Hence, it suffices to show that
\BitP\ belongs to \CH. Consequently we have to show that the
mapping $(\A, \vec{p}) \mapsto N(\A,\vec{p})$ belongs to \CH, where
the Parikh vector $\vec{p}$ is encoded in binary and (by the above remark)
$\A$ is a well-formed DFA. Note that $N(\A,\vec{p})$ is doubly
exponentially bounded in the number of bits of $\vec{p}$ and the size
of the DFA $\A$.  From Lemma~\ref{lemma-BEST->DFA} we get
\[
N(\A,\vec{p}) = \sum_{w \in W(\A,\vec{p})} e(\A^w),
\]
where $e(\A^w)$ can be computed according to
Corollary~\ref{coro:best}. Using Lemma~\ref{lemma-add-mult-CH} and
\ref{lemma-division-CH} we can show that
the mapping $(\A, \vec{p}) \mapsto N(\A,\vec{p})$ belongs to \CH.  For
this, note that for a given mapping $w\colon \Delta \to \mathbb{N}$
one can check in polynomial time whether $w \in
W(\A,\vec{p})$. Moreover, from $w$ and $\A$ one can construct the
edge-weighted multi-graph $\A^w$ in polynomial time. Finally note that
by Tutte's matrix-tree theorem, the number $t(\tilde{\A}_w)$ is a
determinant that can be computed in polynomial time from the edge
weights of $\A_w$, i.e., the values of $w$.
This concludes the proof of Theorem~\ref{thm-parikh-binary-dfa}.

It would be interesting to reduce the upper bound for \BitP\ in
Theorem~\ref{thm-parikh-binary-dfa} further to \textsc{BitSLP} (the
problem of computing a given bit in the number computed by an
arithmetic circuit). But it might be difficult to come up with an
arithmetic circuit for computing $N(\A,\vec{p})$. Note that factorials
appear in the formula from Corollary~\ref{coro:best}.  It is well
known that the existence of polynomial-size arithmetic circuits for
the factorial function implies that integer factoring can be done in
non-uniform polynomial time \cite{Str76}, see also \cite{Burgisser09}.

\section{More on Counting Problems for Parikh Images} \label{sec-further}

In this section, we prove further complexity results for \PIC\ shown
in Table~\ref{tab:main-results} and also the counting problem \countP,
which complement Theorem~\ref{thm-parikh-binary-dfa} (similar results
can also be obtained for \BitP). Due to space constraints, we focus on
DFA-related results; however, most proofs are
deferred to the appendix.

\subsection{Further Results for DFA}
\makeatletter{}In Section~\ref{ssec-best} we proved that \BitP\ and \PIC\ belong to
\CH\ for DFA over a variable alphabet (meaning that the alphabet is
part of the input) and Parikh vectors encoded in binary. Moreover, in
\cite{HK15} it was shown that \PIC\ is \PosSLP-hard for DFA over a
variable alphabet and Parikh vectors encoded in binary (and the proof
shows that in this case \BitP\ is \BitSLP-hard). The variable alphabet
and binary encoding of Parikh vectors are crucial for the proof of the
lower bound. In this section, we show several other results for
DFA when the alphabet is not unary. The results of this section are
collated in the following proposition.
\begin{restatable}{proposition}{propDFAResults}\label{prop:dfa-results}
  For DFA,
  \begin{enumerate}[(i)]
  \item \countP\ (resp.~\PIC) is \sharpL-complete
    (resp.~\PL-complete) for a \emph{fixed} alphabet of size at least two and Parikh
    vectors encoded in \emph{unary};
  \item \countP\ (resp.~\PIC) is \sharpP-complete (resp.~\PP-complete)
    for a \emph{variable} alphabet and Parikh vectors encoded in
    \emph{unary}; and
  \item \PIC\ is \PosMatPow-hard for a \emph{fixed binary} alphabet
    and Parikh vectors encoded in \emph{binary}.
  \end{enumerate}
\end{restatable}
\begin{proof}[Proof of Proposition~\ref{prop:dfa-results}(i)]
  \iftechrep{The lower bound for \sharpL\ follows via a reduction from the
    canonical \sharpL-complete problem of computing the
    number of paths between two nodes in a directed acyclic
    graph~\cite{MV97}, and for the \PL\ lower bound one reduces from the problem
    whether the number of paths from $s$ to $t_0$ is larger than the number of paths from $s$ to $t_1$;
    see the appendix.}{The lower
    bound follows via a reduction from the canonical
    \sharpL-complete problem of computing the number of
    paths between two nodes in a directed acyclic graph~\cite{MV97}.}

  For the upper bound, let $\A$ be a DFA over a fixed alphabet and
  $\vec{p}$ be a Parikh vector encoded in unary.  A non-deterministic
  logspace machine can guess an input word for $\A$ symbol by
  symbol. Thereby, the machine only stores the current state of $\A$
  (which needs logspace) and the binary encoding of the Parikh image
  of the word produced so far. The machine stops when the Parikh
  image reaches the input vector $\vec{p}$ and accepts iff the current
  state is final. Note that since the input
  Parikh vector $\vec{p}$ is encoded in unary notation, all numbers
  that appear in the accumulated Parikh image stored by the machine
  need only logarithmic space. Moreover, since the alphabet has fixed
  size, logarithmic space suffices to store the whole Parikh image.
  The number of accepting computations of the machine is  exactly $N(\A,\vec{p})$, 
  which yields the upper bound for \sharpL\ as well as for \PL.
\end{proof}

\begin{proof}[Proof of Proposition~\ref{prop:dfa-results}(ii)] 
We prove hardness of \countP\ by a reduction from the 
\sharpP-complete counting problem \textsc{\#3SAT}, see
e.g.~\cite[p.~442]{Pap94}: Given a Boolean formula
$\psi(X_1,\ldots,X_n)$ in 3-CNF, compute the number
of satisfying assignments for $\psi$. Let $\psi$ be of the form
$\psi(X_1,\ldots,X_n) = \bigwedge_{1\le i\le k} C_i$, where $C_i$
is the clause $C_i = \ell_{i,1} \vee \ell_{i,2} \vee \ell_{i,3}$ and each $\ell_{i,j}$ is a literal. We define
the alphabet $\Sigma$ used in our reduction as
\begin{multline*}
  \Sigma \ \defeq \{ x_i,\overline{x_i} : 1\le i\le n\} \cup \{ c_i :
  1\le i \le k \} \ \cup \\ 
  \cup \{ d_{i,j} : 1\le i\le k, 0 \le j \leq 2 \}.
\end{multline*}
The informal meaning of the alphabet symbol $x_i$ is that it indicates
that $X_i$ has been set to true, and symmetrically $\overline{x_i}$
indicates that $X_i$ has been set to false. Likewise, $c_i$ indicates
that clause $C_i$ has been set to true. The $d_{i,j}$ will be used as
dummy symbols in order to ensure that the automaton we construct is
deterministic.

\tikzset{initial text={}}
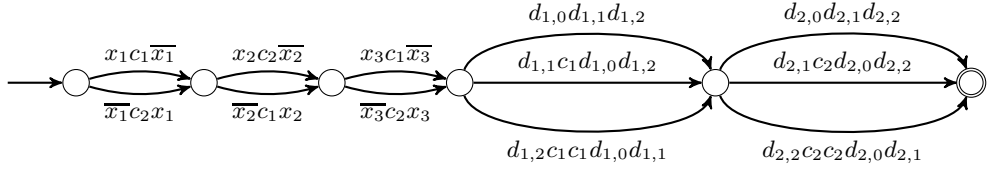
\begin{figure*}[t]
\begin{center}
  \begin{tikzpicture}[xscale=1.7,yscale=0.85,LMC style]
    \node[state, initial] (a) at (0,0) {};
    \node[state] (a1)  at (1,0) {};

    \path[->, bend left] (a) edge node[above] {$x_1c_1\overline{x_1}$} (a1);
    \path[->, bend right] (a) edge node[below] {$\overline{x_1}c_2x_1$} (a1);

    \node[state] (a2) at (2,0) {};

    \path[->, bend left] (a1) edge node[above] {$x_2c_2\overline{x_2}$} (a2);
    \path[->, bend right] (a1) edge node[below] {$\overline{x_2}c_1x_2$} (a2);    

    \node[state] (a3) at (3,0) {};

    \path[->, bend left] (a2) edge node[above] {$x_3 c_1 \overline{x_3}$} (a3);
    \path[->, bend right] (a2) edge node[below] {$\overline{x_3}c_2x_3$} (a3);    

    \node[state] (a4) at (5,0) {};

    \path[->, bend left=80] (a3) edge node[above] {$d_{1,0} d_{1,1}  d_{1,2}$} (a4);
     \path[->]               (a3) edge node[above] {$d_{1,1} c_1 d_{1,0} d_{1,2}$} (a4);
    \path[->, bend right=80] (a3) edge node[below] {$d_{1,2} c_1 c_1 d_{1,0} d_{1,1}$} (a4);    

    \node[state,accepting] (a5) at (7,0) {};

    \path[->, bend left=80] (a4) edge node[above] {$d_{2,0} d_{2,1} d_{2,2}$} (a5);
    \path[->]               (a4) edge node[above] {$d_{2,1} c_2 d_{2,0} d_{2,2}$} (a5);
    \path[->, bend right=80] (a4) edge node[below] {$d_{2,2} c_2 c_2 d_{2,0} d_{2,1}$} (a5);
  \end{tikzpicture}
  \caption{Illustration of the DFA for the reduction from an instance
    $\psi(X_1,X_2,X_3)$ of \textsc{\#3SAT}, where
    $\psi=C_1 \land C_2$ with $C_1=X_1 \lor \neg X_2 \lor X_3$ and
    $C_2=\neg X_1 \vee X_2 \vee \neg X_3$.}
  \label{fig:pp-hardness}
\end{center}
\end{figure*}

Let us now describe how to construct a DFA $\A$ and a Parikh vector
$\vec{p}$ such that $N(\A,\vec{p})$ is the number of satisfying
assignments for $\psi$. The construction of $\A$ is illustrated in
Figure~\ref{fig:pp-hardness}. We construct $\A$ such that it consists
of two phases. In its first phase, $\A$ guesses for every $1 \leq i \leq n$ a valuation of
$X_i$ by producing either (i) $x_i$ followed by all $c_j$ such that $C_j$ contains
$X_i$ followed by $\overline{x_i}$ or (ii) $\overline{x_i}$ followed by all $c_j$ such that $C_j$ contains
$\neg X_i$ followed by $x_i$.
Subsequently in its second phase, $\A$ may non-deterministically produce
at most $2$ additional symbols $c_i$ for every $1\le i\le k$ by
first producing $d_{i,j}$ followed by $j$ letters $c_i$ and all
$d_{i,g}$ such that $j\neq g$. This ensures that $\A$ remains
deterministic.

Now define the Parikh vector $\vec{p}$ over $\Sigma$ as
$\vec{p}(x_i)=\vec{p}(\overline{x_i}) \defeq 1$ for all $1\le i\le n$,
$\vec{p}(c_i)\defeq 3$ for all $1\le i\le j$, and
$\vec{p}(d_{i,j})\defeq 1$ for all $1\le i\le k$ and $0\le j \leq 2$.
The construction of $\A$ ensures that if $\A$ initially guesses a
satisfying assignment of $\psi$ then it can produce a word with Parikh
image $\vec{p}$ in a unique way, since then at least one alphabet
symbol $c_i$ was produced in the first phase of $\A$, and the second
phase of $\A$ can be used in order to make up for missing alphabet
symbols.

In order to show the \sharpP-upper bound for \countP, let $\A$ be a DFA and $\vec{p}$
be a Parikh vector encoded in unary. A non-deterministic
polynomial-time Turing machine can first non-deter\-ministically produce
an arbitrary word $w$ with $\parikh(w)=\vec{p}$. Then, it checks in
polynomial time whether $w\in L(\A)$, in which case it accepts.

The proof that \PIC\ is \PP-complete is similar and can be found in the appendix.
\end{proof}

\begin{remark}
  It is worth mentioning that the above \PP-lower bound together with
  the discussion after Proposition~5 in~\cite{HK15} improves the
  \PP-lower bound of the cost problem by yielding \PP-hardness under
  many-one reductions. In~\cite{HK15}, the cost problem is shown
  \PP-hard via a reduction from the $K$-th largest subset problem,
  which is only known to be \PP-hard under polynomial-time
  \emph{Turing reductions}~\cite{HK16}. In fact, it is not even known
  if this problem is \NP-hard under many-one
  reductions~\cite[p.\ 148]{HS11}.
\end{remark}

\begin{proof}[Proof of Proposition~\ref{prop:dfa-results}(iii)] As stated
in Section~\ref{sec-introduction}, the \PosMatPow\ problem asks, given
a square integer matrix $M \in \Z^{m \times m}$, a linear function $f
\colon \Z^{m \times m} \to \Z$ with integer coefficients, and a positive
integer~$n$, whether $f(M^n) \ge 0$. Unless stated otherwise,
subsequently we assume that all numbers are encoded in binary. Here,
we show that \PIC\ is \PosMatPow-hard for DFA over two-letter
alphabets and Parikh vectors encoded in binary. We first establish two
lemmas that will enable us to prove this proposition. It is well-known
that counting the number of paths in a directed graph corresponds to
matrix powering. In the following lemma, we additionally show that the
application of a linear function can be encoded in such a way as well.

\begin{restatable}{lemma}{lePosMatToMultigraph}\label{lem-PosMat-to-multigraph}
  Given a matrix $M \in \Z^{m \times m}$ and a linear function
  $f\colon \Z^{m \times m} \to \Z$ with integer coefficients, one can compute in
  logspace an edge-weighted multi-graph $G =
  (V,E,s,t,w)$ and $v_0, v^+, v^- \in V$ such that for all $n \in \N$
  we have $f(M^n) = N(G, v_0, v^+, n+2) - N(G, v_0, v^-, n+2)$.
\end{restatable}
\begin{proof}
  Denote by $b_{i,j} \in \Z$ the coefficients of~$f$, i.e., for $i, j
  \in \{1, \ldots, m\}$ let $b_{i,j} \in \Z$ such that for all $A \in
  \Z^{m \times m}$ we have $f(A) = \sum_{i=1}^m \sum_{j=1}^m b_{i,j}
  A_{i,j}$.

  \iftechrep{In the appendix in Lemma~\ref{lem-MatPower-gadget}, we
    show that}{In the full version of this paper we show that} for all
  $i,j \in \{1, \ldots, m\}$ one can compute in logspace an
  edge-weighted multi-graph~$G_{i,j}$ with vertex set~$V_{i,j}$, and
  vertices $v_{i,j}^0, v_{i,j}^+, v_{i,j}^- \in V_{i,j}$ such that for
  all $n \in \N$ we have:
  \begin{equation}\label{eq-lemPosMat-to-multigraph}
    M^n_{i,j} = N(G_{i,j}, v_{i,j}^0, v_{i,j}^+, n) - N(G_{i,j},
    v_{i,j}^0, v_{i,j}^-, n) \iftechrep{}{\notag}
  \end{equation}
  Compute the desired edge-weighted multi-graph~$G$ as follows. For
  each $i,j \in \{1, \ldots, m\}$ include in~$G$ (a fresh copy of) the
  edge-weighted multi-graph $G_{i,j}$.  Further, include in~$G$ fresh
  vertices $v_0, v^+, v^-$, and edges with weight~$1$ from $v_0$
  to~$v_{i,j}^0$, for each $i,j \in \{1, \ldots, m\}$.  Further, for
  each $i,j \in \{1, \ldots, m\}$ with $b_{i,j} > 0$, include in~$G$
  an edge from $v_{i,j}^+$ to~$v^+$ with weight~$b_{i,j}$, and an edge
  from $v_{i,j}^-$ to~$v^-$ with weight~$b_{i,j}$.  Similarly, for
  each $i,j \in \{1, \ldots, m\}$ with $b_{i,j} < 0$, include in~$G$
  an edge from $v_{i,j}^+$ to~$v^-$ with weight~$-b_{i,j}$, and an
  edge from $v_{i,j}^-$ to~$v^+$ with weight~$-b_{i,j}$. \iftechrep{In
    the appendix, we show}{In the full version of this paper, we show}
  that $f(M^n) = N(G, v_0, v^+, n+2) - N(G, v_0, v^-, n+2)$ for all
  $n\in \N$.
\end{proof}
The next lemma shows that one can obtain from an edge-weighted
multi-graph a corresponding DFA such that the number of paths in the
graph corresponds to the number of words with a certain Parikh image
accepted by the DFA. \iftechrep{The lemma is shown in a couple of
  intermeditate steps in the appendix.}{}
\begin{restatable}{lemma}{lemGraphToDFA}\label{lem:graph-to-dfa}
  Let $G = (V,E,s,t,w)$ be an edge-weighted multi-graph, $u,v\in V$
  and $k\in \N$. There exists a logspace-computable DFA $\A$ over a
  two-letter alphabet and a Parikh vector $\vec{p}$ such that
  $N(\A,\vec{p})=N(G,u,v,k)$.
\end{restatable}
In order to prove Proposition~\ref{prop:dfa-results}(iii),
Lemma~\ref{lem-PosMat-to-multigraph} allows us to obtain a multi-graph
$G=(V,E,s,t,w)$ such that $f(M^n)+1 = N(G, v_0, v^+, n+2) - N(G, v_0,
v^-, n+2)$ for vertices $v_0,v^+,v^-\in
V$. Lemma~\ref{lem:graph-to-dfa} yields DFA $\A,\B$ and Parikh vectors
$\vec{p},\vec{p}'$ such that $N(\A,\vec{p})=N(G, v_0, v^+, n+2)$ and
$N(\B,\vec{p}')= N(G, v_0, v^-, n+2)$. In fact, an inspection of the
proof of Lemma~\ref{lem:graph-to-dfa} shows that $\vec{p}=\vec{p}'$.
Hence, $f(M^n)\ge 0$ iff $f(M^n)+1>0$ iff
$N(\A,\vec{p})>N(\B,\vec{p})$.
\end{proof}

\subsection{Further Results for NFA and CFG}\label{section-NFA-CFG}
\makeatletter{}We now show the remaining results for NFA and CFG from
Table~\ref{tab:main-results} when the alphabet is not unary. The
following theorem states upper bounds for \PIC\ and \countP\ for NFA
and CFG.

\begin{restatable}{proposition}{propVariableUpper}\label{prop:nfa-cfg-variable-upper}
    For an alphabet of variable size, \countP\ (resp., \PIC) is in
  \begin{enumerate}[(i)]
  \item \sharpP\ (resp., \PP) for CFG with Parikh vectors encoded in unary;
  \item \sharpPSPACE\ (resp., \PSPACE) for NFA with Parikh vectors encoded in binary; and
  \item \sharpEXP\ (resp., \PEXP) for CFG with Parikh vectors encoded in binary.
  \end{enumerate}
\end{restatable}
\begin{proof}[Proof (sketch)]
In all cases, the proof is a
straightforward adaption of the proof for the upper bounds in
Proposition~\ref{prop:dfa-results}(i), see the appendix.
\end{proof}

\makeatletter{}The following proposition states matching lower bounds for \PIC\ for
the cases considered in Proposition~\ref{prop:nfa-cfg-variable-upper}:
\begin{restatable}{proposition}{propFixedLower}\label{prop:nfa-cfg-fixed-lower}
  For a fixed alphabet of size two, \PIC\ is hard for
  \begin{enumerate}[(i)]
  \item \PP\ for NFA and Parikh vectors encoded in \emph{unary};
  \item \PSPACE\ for NFA and Parikh vectors encoded in \emph{binary}; and
  \item \PEXP\ for CFG and Parikh vectors encoded in \emph{binary}.
  \end{enumerate}
\end{restatable}
\begin{proof}[Proof (sketch)]
  \iftechrep{We only provide the main ideas for the lower bounds, all
    details can be found in the appendix.}{} 
    Let us sketch the proof for (i).
    The proof is based on the fact that those strings (over an alphabet $\Sigma$) that do not encode
    a valid computation (called erroneous below) of a polynomial-space bounded 
    non-deterministic Turing machine $\M$ started on an input $x$ (with $|x|=n$) can be produced by a small NFA \cite{SM73}
    (and this holds also for polynomial-space bounded machines, which is important for (ii)). 
    Suppose the NFA $\A$ generates all words that end
  in an accepting configuration of $\M$, or that are erroneous and end in a
  rejecting configuration. Symmetrically, suppose that $\B$ generates
  all words that are erroneous and end in an accepting configuration,
  or that end in a rejecting configuration. We then
  have that
  $\# (L(\A) \cap \Sigma^{g(n)}) - \# (L(\B) \cap \Sigma^{g(n)})$ 
  equals the difference between the number of accepting paths and rejecting paths of
  $\M$. Here, $g(n)$ is a suitably chosen polynomial. 
  
  Let $h\colon \Sigma^* \to \{0,1\}^*$ be the morphism that maps the $i$-th element of $\Sigma$ (in some enumeration)
  to $0^{i-1} 1 0^{\#\Sigma-i}$. Moreover, let $\A_h$ and $\B_h$ be NFA for $h(L(\A))$ and $h(L(\B))$, respectively, and let 
  $\vec{p}$ be the Parikh vector with $\vec{p}(0)\defeq g(n) \cdot (\#\Sigma-1)$ and
 $\vec{p}(1)\defeq g(n)$. Then $N(\A_h,\vec{p})-N(\B_h,\vec{p}) = \# (L(\A) \cap \Sigma^{g(n)}) - \# (L(\B) \cap \Sigma^{g(n)})$
  equals the difference between the number of accepting paths and rejecting paths of
  $\M$.
    
  The proof for (ii) is similar.
  For (iii) we use the fact that those strings that do not encode
  a valid computation of an exponential-time bounded
  non-deterministic Turing machine started on an input $x$ can be produced by a small CFG
 \cite{HuntRS76}.
\end{proof}
In our construction above, we do not construct an NFA (resp., CFG)
$\A$ and a Parikh vector $\vec{p}$ such that $N(\A, \vec{p})$ is
\emph{exactly} the number of accepting computations of $\M$ on the
given input. This is the reason for not stating hardness for \sharpP\
(resp., \sharpPSPACE\; \sharpEXP) in the above proposition (we could only
show hardness under Turing reductions, but not parsimonious
reductions).

\subsection{Unary alphabets}
\makeatletter{}
A special case of \PIC\ that has been ignored so far is the case of a
unary alphabet. Of course, for a unary alphabet a word is determined
by its length, and a Parikh vector is a single number.  Moreover,
there is not much to count: Either a language $L \subseteq \{a\}^*$
contains no word of length $n$ or exactly one word of length
$n$. Thus, \PIC\ reduces to the question whether for a given length
$n$ (encoded in unary or binary) the word $a^n$ is accepted by $\A$
and rejected by $\B$. In this section we clarify the complexity of
this problem for (i) unary DFA, NFA, and CFG, and (ii) lengths encoded
in unary and binary. In the case of lengths encoded in binary,
\PIC\ is tightly connected to the {\em compressed word problem}: Given
a unary DFA (resp., NFA, CFG) $\A$ and a number~$n$ in binary
encoding, determine if $n \in L(\A)$. In particular, if this problem
belongs to a complexity class that is closed under complement
(e.g. $\LOGSPACE$, $\NL$, $\P$), then \PIC\ belongs to the same class.

\begin{restatable}{proposition}{propDFAUnaryUpper}\label{prop:dfa-unary-upper}
  For unary alphabets, \PIC\ is
    \begin{enumerate}[(i)]
  \item in \LOGSPACE\ for DFA with Parikh vectors encoded in \emph{binary};
  \item \NL-complete for NFA irrespective of the encoding of the
    Parikh vector; and
  \item \P-complete and \DP-complete for CFG with Parikh vectors encoded
    in unary and binary, respectively.
  \end{enumerate}
\end{restatable}
All proofs are given in the appendix. Perhaps most interestingly, for Part~(ii) we apply a recent result by Sawa~\cite{Sawa13} on
arithmetic progressions in unary NFA in order to show that the
compressed word problem for NFA is in \NL. The \DP-lower bound in
Part~(iii) can be shown via a reduction from a variant of the
classical subset sum problem, exploiting the fact that CFGs can
generate exponentially large numbers and hence, informally speaking,
encode numbers in binary.

\section{Perspectives}
\makeatletter{}We studied and established a close connection between the cost problem
and \PIC, a natural problem language-theoretic problem that
generalizes counting words in DFA, NFA and CFG. The main results of
this paper are \CH-membership of \PIC\ and the cost problem. It is
clear that this upper bound carries over to a multi-dimensional
version of the cost problem where cost formulas are linear
inequalities of the form $x_1\le b_1\wedge \cdots \wedge x_n\le
b_n$. A challenging open problem seems to be whether decidability can
be retained when arbitrary Presburger formulas are allowed as cost
formulas.

\bibliographystyle{plain}
\bibliography{bibliography}

\iftechrep{
\clearpage
\newpage
\appendix
\makeatletter{}\section{Missing Proofs from Section~\ref{sec-further}}
Here, we provide details of the omitted proofs from the main part.

\begin{restatable}{lemma}{lemMatPowerGadget}\label{lem-MatPower-gadget}
  Given a matrix $M \in \Z^{m \times m}$, and $i,j \in \{1, \ldots,
  m\}$, one can compute in logspace an edge-weighted multi-graph
  $G = (V,E,s,t,w)$ and $v_i^+, v_j^+, v_j^- \in V$ such that for all
  $n \in \N$ we have $(M^n)_{i,j} = N(G, v_i^+, v_j^+, n) - N(G,
  v_i^+, v_j^-, n)$.
\end{restatable}
\begin{proof}
In the following we write $M^n_{i,j}$ to mean $(M^n)_{i,j}$.
Define an edge-weighted multi-graph $G = (V,E,s,t,w)$ as follows.
Let $V = \{v_k^+, v_k^- : 1 \le k \le m\}$.
For all $k,\ell \in \{1, \ldots, m\}$, if $M_{k,\ell} > 0$ then include in~$E$
an edge~$e$ from~$v_k^+$ to~$v_\ell^+$ with $w(e) = M_{k,\ell}$, and 
an edge~$e$ from~$v_k^-$ to~$v_\ell^-$ with $w(e) = M_{k,\ell}$.
Similarly, if $M_{k,\ell} < 0$ then include in~$E$
an edge~$e$ from~$v_k^+$ to~$v_\ell^-$ with $w(e) = -M_{k,\ell}$, and 
an edge~$e$ from~$v_k^-$ to~$v_\ell^+$ with $w(e) = -M_{k,\ell}$.
We prove by induction on~$n$ that we have for all $k, \ell \in \{1, \ldots, m\}$:
\[
M^n_{k,\ell} = N(G, v_k^+, v_\ell^+, n) - N(G, v_k^+, v_\ell^-, n)
\]
Note that this implies the statement of the lemma.
For the induction base, let $n=0$.
If $k = \ell$ then $M^n_{k,\ell} = 1$, $N(G, v_k^+, v_\ell^+, 0) = 1$, and $N(G, v_k^+, v_\ell^-, 0) = 0$.
If $k \ne \ell$ then $M^n_{k,\ell} = 0 = N(G, v_k^+, v_\ell^+, 0) = N(G, v_k^+, v_\ell^-, 0)$.
For the inductive step, let $n \in \N$ and suppose $M^n_{k,\ell} = N(G, v_k^+, v_\ell^+, n) - N(G, v_k^+, v_\ell^-, n)$ for all $k, \ell$.
For $s \in \{1, \ldots, m\}$ write
$I^+(s) := \{\ell \in \{1, \ldots, m\}: M_{\ell,s} > 0\}$ and
$I^-(s) := \{\ell \in \{1, \ldots, m\}: M_{\ell,s} < 0\}$.
For $v, v', v'' \in V$ write 
$\tilde{N}(G, v, v', v'', n+1)$ for the number of paths in~$\tilde{G}$ (the unweighted version of $G$) from $v$ to~$v''$ of length~$n+1$ such that $v'$ is the vertex visited after $n$ steps.
We have for all $k, s \in \{1, \ldots, m\}$:
\begin{eqnarray*}
M^{n+1}_{k,s} 
& = &  \sum_{\ell=1}^m M^n_{k,\ell} M_{\ell,s} \\
& \stackrel{\text{(ind. hyp.)}}{=} &  \sum_{\ell=1}^m N(G, v_k^+, v_\ell^+, n) M_{\ell,s} \ - \\
& &  \sum_{\ell=1}^m N(G, v_k^+, v_\ell^-, n) M_{\ell,s} \\
& = &  \sum_{\ell \in I^+(s)} N(G, v_k^+, v_\ell^+, n) M_{\ell,s} \ + \\
& &  \sum_{\ell \in I^-(s)} N(G, v_k^+, v_\ell^-, n) (-M_{\ell,s}) \ - \\
& & \sum_{\ell \in I^+(s)} N(G, v_k^+, v_\ell^-, n) M_{\ell,s} \ - \\
& & \sum_{\ell \in I^-(s)} N(G, v_k^+, v_\ell^+, n) (-M_{\ell,s}) \\
& = &  \sum_{\ell \in I^+(s)} \tilde{N}(G, v_k^+, v_\ell^+, v_s^+, n+1) \ + \\
& & \sum_{\ell \in I^-(s)} \tilde{N}(G, v_k^+, v_\ell^-, v_s^+, n+1) \ - \\
& & \sum_{\ell \in I^+(s)} \tilde{N}(G, v_k^+, v_\ell^-, v_s^-, n+1) \ - \\
& & \sum_{\ell \in I^-(s)} \tilde{N}(G, v_k^+, v_\ell^+, v_s^-, n+1) \\
& = &   N(G, v_k^+, v_s^+, n+1) - N(G, v_k^+, v_s^-, n+1)
\end{eqnarray*}
This completes the induction proof.
\end{proof}

We prove the following lemma from the main text:

\lePosMatToMultigraph*
\begin{proof}
  We deferred showing that $f(M^n) = N(G, v_0, v^+, n+2) - N(G, v_0,
  v^-, n+2)$ for all $n \in \N$. Indeed, any path of length~$n+2$ from
  $v_0$ to~$v^+$ must start with an edge from $v_0$ to~$v_{i,j}^0$ for
  some $i,j$, continue with a path of length~$n$ from $v_{i,j}^0$ to
  either $v_{i,j}^+$ or~$v_{i,j}^-$, and finish with an edge to~$v^+$.
  Hence, writing $I^+ := \{(i,j) : 1 \le i,j \le m, \ b_{i,j} > 0\}$
  and $I^- := \{(i,j) : 1 \le i,j \le m, \ b_{i,j} < 0\}$ we have
\begin{multline*}
  N(G, v_0, v^+, n+2) = 
  \sum_{(i,j) \in I^+} N(G, v_{i,j}^0, v_{i,j}^+, n) \cdot b_{i,j} \\
   + \sum_{(i,j) \in I^-} N(G, v_{i,j}^0, v_{i,j}^-, n) \cdot (-b_{i,j}).
\end{multline*}
Similarly we have:
\begin{multline*}
  N(G, v_0, v^-, n+2) = 
  \sum_{(i,j) \in I^+} N(G, v_{i,j}^0, v_{i,j}^-, n) \cdot b_{i,j} \\
  + \sum_{(i,j) \in I^-} N(G, v_{i,j}^0, v_{i,j}^+, n) \cdot (-b_{i,j}).
\end{multline*}
Hence we have:
\begin{eqnarray*}
  f(M^n) & = & \sum_{i=1}^m \sum_{j=1}^m M^n_{i,j} \cdot b_{i,j}  \\
  & \stackrel{\text{\eqref{eq-lemPosMat-to-multigraph}}}{=} &  \sum_{i=1}^m \sum_{j=1}^m N(G, v_{i,j}^0, v_{i,j}^+, n) \cdot b_{i,j} \ - \\ 
  & & \sum_{i=1}^m \sum_{j=1}^m N(G, v_{i,j}^0, v_{i,j}^-, n) \cdot b_{i,j} \\ 
  & = &  N(G, v_0, v^+, n+2) - N(G, v_0, v^-, n+2) 
\end{eqnarray*}
\end{proof}

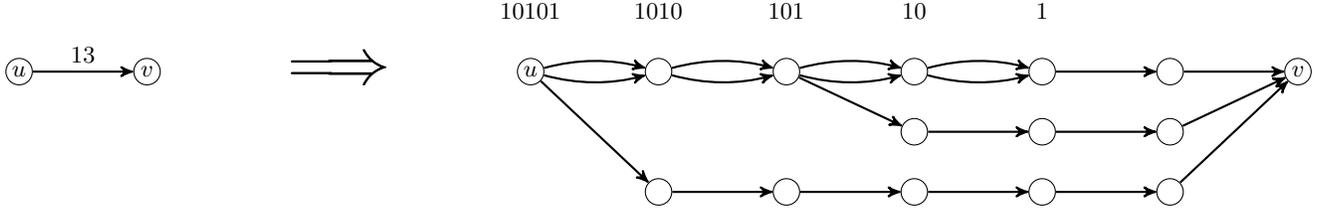
\begin{figure*}[t]
\begin{center}
\begin{tikzpicture}[xscale=1.7,yscale=0.8,LMC style]
\node[state] (a) at (-1,0) {$u$};
\node[state] (b) at (0,0) {$v$};
\path[->] (a) edge node[above] {$13$} (b);
\node (arrow) at (1.5,0) {\Huge $\Longrightarrow$};

\node at (3,1) {$10101$};
\node at (4,1) {$1010$};
\node at (5,1) {$101$};
\node at (6,1) {$10$};
\node at (7,1) {$1$};

\node[state] (u) at (3,0) {$u$};
\node[state] (u1) at (4,0) {};
\node[state] (u2) at (5,0) {};
\node[state] (u3) at (6,0) {};
\node[state] (u4) at (7,0) {};
\node[state] (u5) at (8,0) {};
\node[state] (u6) at (9,0) {$v$};

\node[state] (v3) at (6,-1) {};
\node[state] (v4) at (7,-1) {};
\node[state] (v5) at (8,-1) {};

\node[state] (w1) at (4,-2) {};
\node[state] (w2) at (5,-2) {};
\node[state] (w3) at (6,-2) {};
\node[state] (w4) at (7,-2) {};
\node[state] (w5) at (8,-2) {};

\path[->, bend left] (u) edge (u1);
\path[->, bend right] (u) edge (u1);
\path[->, bend left] (u1) edge (u2);
\path[->, bend right] (u1) edge (u2);
\path[->, bend left] (u2) edge (u3);
\path[->, bend right] (u2) edge (u3);
\path[->, bend left] (u3) edge (u4);
\path[->, bend right] (u3) edge (u4);
\path[->] (u4) edge (u5);
\path[->] (u5) edge (u6);

\path[->] (u2) edge (v3);
\path[->] (v3) edge (v4);
\path[->] (v4) edge (v5);
\path[->] (v5) edge (u6);

\path[->] (u) edge (w1);
\path[->] (w1) edge (w2);
\path[->] (w2) edge (w3);
\path[->] (w3) edge (w4);
\path[->] (w4) edge (w5);
\path[->] (w5) edge (u6);

\end{tikzpicture}
\end{center}
\caption{Illustration of the construction of the unweighted multi-graph from Lemma~\ref{lem-PosMat-multigraph-to-graph}.
We assume $k=6$.
The binary representation of~$13$ is $10101$.
The binary numbers over the nodes on the right hand side correspond to $w$-values that occur during the construction, but are not part of the output.
Each binary number over a node indicates the number of paths to~$v$.
}
\label{fig:PosMat-multigraph-to-graph}
\end{figure*}

\begin{restatable}{lemma}{PosMatMultgraphToGraph}\label{lem-PosMat-multigraph-to-graph}
  Given an edge-weighted multi-graph $G = (V,E,s,t,w)$ (with $w$ in
  binary), $v_0, v_1 \in V$ and a number $k \in \N$ in unary such that
  $k \ge 1 + \max_{e \in E} \lfloor \log_2 w(e) \rfloor$, one can
  compute in logspace an unweighted multi-graph $G' =
  (V',E',s',t')$ with $V' \supseteq V$ such that for all $n \in \N$ we
  have $N(G,v_0,v_1,n) = N(G',v_0,v_1,n \cdot k)$.
\end{restatable}
\begin{proof}
  Note that $k$ is at least the size of the binary representation of
  the largest weight in $G$.  Define a mapping $b \colon E \to \N$
  with $b(e) = k$ for all $e \in E$.  Define~$G'$ so that it is
  obtained from~$G$ as follows.  Let $e \in E$ with $b(e) > 1$.  If
  $w(e) = 1$ then replace~$e$ by a fresh path of length~$b(e)$ (with
  $w(e') = b(e') = 1$ for all edges~$e'$ on that path).  If $w(e) = 2
  j$ for some $j \in \N$ then introduce a fresh vertex $v$ and two
  fresh edges $e_1, e_2$ from $s(e)$ to~$v$ with $b(e_1) = b(e_2) =
  w(e) = 1$ and another fresh edge $e_3$ from $v$ to~$t(e)$ with
  $b(e_3) = b(e) - 1$ and $w(e_3) = j$.  Finally, if $w(e) = 2 j + 1$
  for some $j \in \N$ then proceed similarly, but additionally
  introduce fresh vertices that create a new path of length~$b(e)$
  from $s(e)$ to~$t(e)$ (with $w(e') = b(e') = 1$ for all edges~$e'$
  on that path).  By this construction, every edge~$e$ is eventually
  replaced by $w(e)$ paths of length~$k$.  The construction is
  illustrated in Figure~\ref{fig:PosMat-multigraph-to-graph}.

  For the logspace claim, note that it is not necessary to store the
  whole graph for this construction. The binary representation of~$k$
  has logarithmic size and can be stored, and a copy of~$k$ can be
  counted down, keeping track of the $b$-values in the construction.
  The edges can be dealt with one by one.  It is not necessary to
  store the values $w(e') = j$ for the created fresh edges; rather
  those values can be derived from the binary representation of the
  original weight~$w(e)$ and the current $b$-value (acting as a
  ``pointer'' into the binary representation of~$w(e)$).
\end{proof}

\begin{restatable}{lemma}{lemAddOnePath} \label{lem-add-one-path}
  Given an unweighted multi-graph $G = (V, E, s, t)$ and $v_0, v_1 \in
  V$, one can compute in logspace unweighted multi-graphs
  $G_0 = (V_0, E_0, s_0, t_0)$ and $G_1 = (V_1, E_1, s_1, t_1)$ with
  $V_0 \supseteq V$ and $V_1 \supseteq V$ such that for all $n \in \N$
  we have $N(G_0,v_0,v_1,n+2) = N(G,v_0,v_1,n)$ and
  $N(G_1,v_0,v_1,n+2) = N(G,v_0,v_1,n) + 1$.
\end{restatable}

\begin{proof}
For $G_0$ 
redirect all edges adjacent to~$v_0$ to a fresh vertex~$v_0^*$, and similarly 
redirect all edges adjacent to~$v_1$ to a fresh vertex~$v_1^*$.
Then add an edge from $v_0$ to~$v_0^*$, and an edge from $v_1^*$ to~$v_1$.

For~$G_1$ do the same, and in addition add a fresh vertex~$v$, and add edges from $v_0$ to~$v$, and from $v$ to~$v_1$, and a loop on~$v$. This adds a path from $v_0$ to~$v_1$ of length $n+2$.
\end{proof}

\begin{lemma} \label{lem-graph-to-DFA}
Given an unweighted multi-graph $G = (V, E, s, t)$,  $v_0, v_1 \in V$ and a number~$d$ in unary so that $d$ is at least the maximal out-degree of any node in~$G$, one can compute 
in logspace a DFA $\A=(Q,\Sigma,q_0,F,\Delta)$ with $\Sigma = \{a, b\}$ such that for all $n \in \N$ we have $N(G,v_0,v_1,n) = N(\A,\vec{p})$ where $\vec{p}(a) = n$ and $\vec{p}(b) = n \cdot (d-1)$.
\end{lemma}
\begin{proof}
Define~$\A$ so that $Q \supseteq V$, $q_0 = v_0$, and $F = \{v_1\}$.
Include states and transitions in~$\A$ so that for every edge~$e$ (from $v$ to~$v'$, say) in~$G$ there is a run from $v$ to~$v'$ in~$\A$ of length~$d$ so that exactly one transition on this run is labelled with~$a$, and the other $d-1$ transitions are labelled with~$b$.
Importantly, each edge~$e$ is associated to exactly one such run.
The construction is illustrated in Figure~\ref{fig:graph-to-DFA}.
The DFA~$\A$ is of quadratic size and can be computed in logspace.
It follows from the construction that any path of length~$n$ in~$G$ corresponds to a run of length $n \cdot d$ in~$\A$, with $n$ transitions labelled with~$a$, and $n \cdot (d-1)$ transitions labelled with~$b$.
This implies the statement of the lemma.
\end{proof}

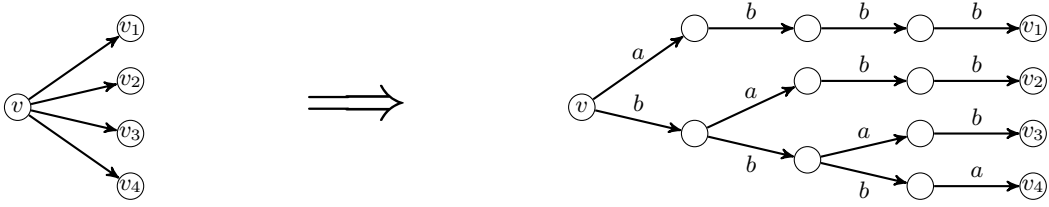
\begin{figure*}[t]
\begin{center}
\begin{tikzpicture}[xscale=1.5,yscale=0.7,LMC style]
\node[state] (v) at (-2,0) {$v$};
\node[state] (v1) at (-1,1.5) {$v_1$};
\node[state] (v2) at (-1,0.5) {$v_2$};
\node[state] (v3) at (-1,-0.5) {$v_3$};
\node[state] (v4) at (-1,-1.5) {$v_4$};
\path[->] (v) edge (v1);
\path[->] (v) edge (v2);
\path[->] (v) edge (v3);
\path[->] (v) edge (v4);
\node (arrow) at (1,0) {\Huge $\Longrightarrow$};
\node[state] (w) at (3,0) {$v$};
\node[state] (w1) at (7,1.5) {$v_1$};
\node[state] (w2) at (7,0.5) {$v_2$};
\node[state] (w3) at (7,-0.5) {$v_3$};
\node[state] (w4) at (7,-1.5) {$v_4$};
\node[state] (w11) at (4,1.5) {};
\node[state] (w21) at (5,1.5) {};
\node[state] (w31) at (6,1.5) {};
\node[state] (w12) at (4,-0.5) {};
\node[state] (w22) at (5,+0.5) {};
\node[state] (w32) at (6,+0.5) {};
\node[state] (w23) at (5,-1.0) {};
\node[state] (w33) at (6,-0.5) {};
\node[state] (w34) at (6,-1.5) {};

\path[->] (w) edge node[above] {$a$} (w11);
\path[->] (w) edge node[above] {$b$} (w12);
\path[->] (w11) edge node[above] {$b$} (w21);
\path[->] (w21) edge node[above] {$b$} (w31);
\path[->] (w31) edge node[above] {$b$} (w1);
\path[->] (w12) edge node[above] {$a$} (w22);
\path[->] (w22) edge node[above] {$b$} (w32);
\path[->] (w32) edge node[above] {$b$} (w2);
\path[->] (w12) edge node[below] {$b$} (w23);
\path[->] (w23) edge node[above] {$a$} (w33);
\path[->] (w23) edge node[below] {$b$} (w34);
\path[->] (w33) edge node[above] {$b$} (w3);
\path[->] (w34) edge node[above] {$a$} (w4);

\end{tikzpicture}
\end{center}
\caption{Illustration of the construction of the DFA from Lemma~\ref{lem-graph-to-DFA}.
We assume $d=4$.}
\label{fig:graph-to-DFA}
\end{figure*}

We prove the following proposition from the main text:

\propDFAResults*
\begin{proof}[Further details to Part~(i)]
  We deferred showing the  lower bounds from
  Part~(i). The classical \sharpL-hard counting problem is the
  computation of the number of paths between a source node $s$ and a
  target node $t$ in a directed acyclic graph $G=(V,E)$
  ~\cite{MV97}. Let $m=\#V$ and $d$ be the maximum out-degree of
  $G$. Let $G_t$ be the multi-graph obtained by adding a loop at node
  $t$. Then, since every path in $G$ from $s$ to $t$ has length at
  most $m$, the number of paths from $s$ to $t$ in $G$ is
  $N(G_t,s,t,m)$.  Now let $\mathcal{A}=(Q,\{a,b\},s,\{t\},\Delta)$ be
  the DFA obtained from Lemma~\ref{lem-graph-to-DFA}. Hence, we have $N(G_t,s,t,m) = N(\A,\vec{p})$,
  where $\vec{p}(a)\defeq m$ and $\vec{p}(b)\defeq m\cdot (d-1)$.
  
  \PL-hardness for \PIC\ can be shown by a reduction from the following problem: 
  Given a directed acyclic graph $G=(V,E)$ and three nodes $s, t_0, t_1$, is the 
  number of paths from $s$ to $t_0$ larger than the number of paths from $s$ to $t_1$.
  We then apply the above reduction to the triples $(G,s,t_0)$ and $(G,s,t_1)$. This
  yields pairs $(\A,\vec{p})$ and $(\B,\vec{q})$, where $\A$ and $\B$ are DFA over the alphabet $\{a,b\}$ and 
  $\vec{p}$ and $\vec{q}$ are unary coded Parikh vectors,
  such that $N(\A,\vec{p})$ (resp.~$N(\B,\vec{q})$) is the number of paths in $G$ from $s$ to $t_0$ (resp.~$t_1$).
  Finally, note that the reduction yields $\vec{p} = \vec{q}$.
  This concludes the proof.
\end{proof}
\begin{proof}[Further details to Part~(ii)]
In the main part we gave a reduction from \textsc{\#3SAT} to
\countP. We now prove \PP-hardness of \PIC\ by reusing this reduction.
It is \PP-complete to check for two given 3-CNF formulas $F$ and $G$
whether $F$ has more satisfying assignments than $G$.\footnote{Note
  that every language in \PP\ is of the form $\{ x : f(x) > g(x) \}$
  for two \sharpP-functions $f$ and $g$. From $x$ one can construct
  two 3-CNF formulas $F$ and $G$ such that the number of satisfying
  assignments of $F$ (resp.~$G$) is $f(x)$ (resp.~$g(x)$).}
W.l.o.g. one can assume that $F$ and $G$ use the same set of Boolean
variables (we can add dummy variables if necessary) and the same
number of clauses (we can duplicate clauses if necessary). We now
apply to $F$ and $G$ the reduction from the \sharpP-hardness proof of
\countP\ from the main part of the paper.  We obtain two pairs
$(\A,\vec{p})$ and $(\B,\vec{q})$, where $\A$ and $\B$ are DFA and
$\vec{p}$ and $\vec{q}$ are Parikh vectors encoded in binary, such
that $N(\A,\vec{p})$ (resp.~$N(\B,\vec{q})$) is the number of
satisfying assignments of $F$ (resp.~$G$). But since $F$ and $G$ are
3-CNF formulas with the same variables and the same number of clauses,
it follows that $\A$ and $\B$ are DFA over the same alphabet and
$\vec{p} = \vec{q}$. This concludes the proof.
\end{proof}

\begin{proof}[Further details to Part~(iii)]
  The above lemmas enable us to prove Part~(iii). Consider an instance
  of \PosMatPow, i.e., a square integer matrix $M \in \Z^{m \times
    m}$, a linear function $f\colon \Z^{m \times m} \to \Z$ with
  integer coefficients, and a positive integer~$n$.  Using
  Lemma~\ref{lem-PosMat-to-multigraph} we can compute in logspace
  edge-weighted multi-graphs $G_+$ with vertices $v^+_0, v^+$ and
  $G_-$ with vertices $v^-_0, v^-$ such that
 \[
  f(M^n) = N(G_+, v^+_0, v^+, n+2) - N(G_-, v^-_0, v^-, n+2)\,.
 \]
Let $k = 1 + \max_{e \in E} \lfloor \log_2 w(e) \rfloor$, where $E$ is the union of the edge sets of $G_+$ and $G_-$.
Using Lemma~\ref{lem-PosMat-multigraph-to-graph} we can compute unweighted multi-graphs $G_+', G_-'$ such that
\begin{align*}
 N(G_+, v^+_0, v^+, n+2) &= N(G_+', v^+_0, v^+, (n+2) \cdot k) \qquad \text{and} \\
 N(G_-, v^-_0, v^-, n+2) &= N(G_-', v^-_0, v^-, (n+2) \cdot k)\,.
\end{align*}
Hence,
\[
 f(M^n) = N(G_+', v^+_0, v^+, k \cdot (n+2)) - N(G_-', v^-_0, v^-, k \cdot (n+2)) \,.
\]
Using Lemma~\ref{lem-add-one-path} we can compute unweighted multi-graphs $G_+'', G_-''$ such that
\begin{align*}
 1 + N(G_+', v^+_0, v^+, (n+2) \cdot k) &= N(G_+'', v^+_0, v^+, (n+2) \cdot k + 2) \\  N(G_-', v^-_0, v^-, (n+2) \cdot k) &= N(G_-'', v^-_0, v^-, (n+2) \cdot k + 2) \,.
\end{align*}
Hence,
\begin{align*}
 f(M^n) + 1 & = N(G_+'', v^+_0, v^+, (n+2) \cdot k + 2) \\
            & \ \mbox{} - N(G_-'', v^-_0, v^-, (n+2) \cdot k + 2)\,.
\end{align*}
Let $d$ denote the maximal out-degree of any node in $G_+''$
or~$G_-''$.  Let $\vec{p} \colon \{a,b\} \to \N$ with $\vec{p}(a) =
(n+2) \cdot k + 2$ and $\vec{p}(b) = ((n+2) \cdot k + 2) \cdot (d-1)$.
Using Lemma~\ref{lem-graph-to-DFA} we can compute DFA $\A, \B$ over
the alphabet~$\{a,b\}$ such that
\begin{align*}
N(G_+'', v^+_0, v^+, (n+2) \cdot k + 2) &= N(\A,\vec{p}) \qquad \text{and} \\
N(G_-'', v^-_0, v^-, (n+2) \cdot k + 2) &= N(\B,\vec{p})\,.
\end{align*}
Hence,
\[
 f(M^n) + 1 = N(\A,\vec{p}) - N(\B,\vec{p}) \,.
\]
So $f(M^n) \ge 0$ if and only if $f(M^n) + 1 > 0$ if and only if $N(\A,\vec{p}) > N(\B,\vec{p})$.
All mentioned computations can be performed in logspace.
\end{proof}

We prove the following proposition from the main text:

\propVariableUpper*

\begin{proof}
It suffices to show the statements for the \#-classes.
Let us consider (i). Let $\A$ be the input CFG and
$\vec{p}$ be the input Parikh vector, which is encoded in unary
notation.  A non-deterministic polynomial-time machine can first
non-deterministically produce an arbitrary word $w$ with
$\parikh(w)=\vec{p}$. Then, it checks in polynomial time whether $w
\in L(\A)$, in which case it accepts.

For (ii) we argue as in the proof of the \sharpL\ upper bound from Proposition~\ref{prop:dfa-results}(i), except
that we simulate the deterministic power set automaton for the input NFA. For this, polynomial space is needed.
Moreover, also the accumulated Parikh image of the prefix guessed so far needs polynomial space.

Finally, for (iii) we can argue as in (i) by using a non-deterministic
exponential time machine.
\end{proof}

We prove the following proposition from the main text:

\propFixedLower*
\begin{proof}
We show those hardness results by developing a generic approach that
only requires minor modifications in each case. In general, we
simulate computations of space-bounded Turing machines as words of NFA
and CFG, respectively. Let $\M = (Q,\Gamma,\Delta)$ be a Turing
machine that uses $f(n) \geq n$ tape cells during a computation on an input
of length $n$, where $\Delta \subseteq Q \times \Gamma \times
Q \times \Gamma \times \{ \leftarrow, \rightarrow
\}$. Unsurprisingly, $(q,a,q',a',d)\in \Delta$ means that if $\M$ is
in control state $q$ reading $a$ at the current head position then
$\M$ can change its control state to $q'$ while writing $a'$ and
subsequently moving its head in direction $d$. Without loss of
generality we may assume that $\M$ uses alphabet symbols $0,1\in
\Gamma$ on its working tape as well as $\triangleright, \triangleleft
\in \Gamma$ as delimiters indicating the left respectively right
boundary of the working tape. Consequently, a \emph{valid
  configuration of $\M$} is a string of length $f(n)+1$ over the
alphabet $\Sigma \defeq \{ \triangleright, 0, 1, \triangleleft \} \cup
Q$ from the language
\[
(\triangleright\cdot \{0,1\}^*\cdot Q\cdot \{0,1\}^*\cdot
\triangleleft) \cup (Q\cdot \triangleright\cdot \{0,1\}^* \cdot
\triangleleft).
\]
With no loss of generality, we moreover make the following assumptions
on $\M$: (i) the \emph{initial configuration} of $\M$ when run on
$x\in \{0,1\}^*$ is the string  $\triangleright \cdot q_0 \cdot x \cdot 0^{f(n)-n} \cdot
\triangleleft$ for some designated control state
$q_0\in Q$; (ii) delimiters are never changed by $\M$ and $\M$ adds no
further delimiter symbols during its computation; and (iii) the
\emph{accepting configuration} of $\M$ is $\triangleright \cdot q_f \cdot
0^{f(n)} \cdot \triangleleft$ for some designated control state $q_f\in Q$,
and any other configuration is \emph{rejecting}. If $\M$ is $f(n)$-time
bounded (and thus $f(n)$-space bounded) then we assume that all
computation paths of $\M$ are of length $f(n)$.

We now turn towards proving
Proposition~\ref{prop:nfa-cfg-fixed-lower}(i). Assume that $\M$ is an
$f(n)$-time bounded non-deterministic Turing machine for a polynomial $f(n)$ and $x
\in \{0,1\}^n$ is an input for $\M$ of length $n$.
We encode
computations of $\M$ as strings over the extended alphabet
$\Sigma_\$\defeq \Sigma \cup \{\$\}$ where $\$$ serves as a separator
between consecutive configurations. Hence a valid computation is
encoded as a string in the language $(\Sigma^{f(n)+1} \cdot
\$)^*$. Let $L_{\text{val}} \subseteq
\Sigma_\$^*$ be the language consisting of
all strings that encode valid computations of
$\M$ when run on $x \in \{0,1\}^n$, and let $L_{\text{inv}}\defeq
\Sigma_\$^* \setminus L_{\text{val}}$.  It is shown in \cite{SM73}
that an NFA for $L_{\text{inv}}$ can be constructed in logspace from the input $x$.
Moreover, let $L_\text{acc}\subseteq
\Sigma^{f(n)+1} \cdot \$$ be the singleton language containing the string
representing the accepting configuration, and let
$L_{\text{rej}}\subseteq \Sigma^{f(n)+1} \cdot \$$ be the set of all
encodings of rejecting configurations. It is straightforward to construct NFA for these
sets in logspace. Hence, we can also construct in logspace
NFA $\A$ and $\B$ such that
\[
L(\A) =  (\Sigma_\$^* \cdot L_\text{acc}) \cup (L_\text{inv} \cap (\Sigma_\$^* \cdot L_\text{rej})),
\]
i.e., $L(\A)$ contains those strings that end in an accepting
configuration and those strings not representing a valid computation
that end in a rejecting configuration, and likewise $\B$ is such that
\[
L(\B) = (  L_\text{inv} \cap (\Sigma^*_\$ \cdot L_\text{acc}) ) \cup
   ( \Sigma_\$^* \cdot L_{\text{rej}} ).
\]
Set $g(n)\defeq f(n) \cdot(f(n) + 2)$. We then get 
\begin{align*}
  & \# (L(\A) \cap \Sigma_\$^{g(n)}) - \# (L(\B) \cap \Sigma_\$^{g(n)})\\
 = \ \ & \# (\Sigma_\$^* \cdot L_{\text{acc}} \cap \Sigma_\$^{g(n)}) \ + \\
        & \# (L_\text{inv} \cap \Sigma_\$^* \cdot L_\text{rej} \cap \Sigma_\$^{g(n)} ) \ - \\
        & \# (L_\text{inv} \cap \Sigma^*_\$ \cdot L_\text{acc} \cap \Sigma_\$^{g(n)}) \ - \\
        & \# (\Sigma_\$^* \cdot L_{\text{rej}} \cap  \Sigma_\$^{g(n)}) \\
= \ \ & \# (L_{\text{val}} \cap \Sigma_\$^* \cdot L_{\text{acc}} \cap \Sigma_\$^{g(n)} ) \  -\\
       & \# (L_{\text{val}} \cap \Sigma_\$^* \cdot L_{\text{rej}} \cap \Sigma_\$^{g(n)} ).
\end{align*}
Consequently, $\# (L(\A) \cap \Sigma_\$^{g(n)}) > \# (L(\B) \cap
\Sigma_\$^{g(n)})$ if and only if $\M$ accepts $x$. Let $h\colon
\Sigma_\$ \to (0^*\cdot 1 \cdot 0^* \cap \{0,1\}^{\# \Sigma_\$})$ be a
bijection that maps every symbol in $\Sigma_\$$ to a string consisting
of exactly one symbol $1$ and $\#\Sigma_\$-1$ symbols $0$. In
particular, note that $\parikh(h(a))=\parikh(h(b))$ for all $a,b\in
\Sigma_\$$. Moreover, let $\A_h$ and $\B_h$ be the NFA recognising the
homomorphic images of $L(\A)$ and $L(\B)$ under $h$. We now have
\begin{align*}
  \# (L(\A) \cap \Sigma_\$^{g(n)}) = \# (L(\A_h)\cap \{0,1\}^{g(n)\cdot \#\Sigma_\$}) =
  N(\A_h,\vec{p}),
\end{align*}
where $\vec{p}(0)\defeq g(n) \cdot (\#\Sigma_\$-1)$ and
$\vec{p}(1)\defeq g(n)$, and analogously $\# (L(\B) \cap
\Sigma_\$^{g(n)}) = N(\B_h,\vec{p})$. Hence,
\begin{align*}
       & N(\A_h,\vec{p}) > N(\B_h,\vec{p})\\
  \iff & \# (L(\A) \cap \Sigma_\$^{g(n)}) > \# (L(\B) \cap
  \Sigma_\$^{g(n)})\\
  \iff & \M \text{ accepts } x.
\end{align*}
This concludes the proof of
Proposition~\ref{prop:nfa-cfg-fixed-lower}(i).

In order to prove Proposition~\ref{prop:nfa-cfg-fixed-lower}(ii), let
$\M$ be an $f(n)$-space bounded Turing machine for a polynomial
$f(n)$. In particular, with no loss of generality we can assume that
if $\M$ has an accepting run then it has one which accepts after
$2^{f(n)}$ steps. All we have to do in order to prove
$\PSPACE$-hardness of \PIC\ is to make two adjustments to the
construction given for
Proposition~\ref{prop:nfa-cfg-fixed-lower}(i). First, we redefine $\A$
and $\B$ such that they recognise the languages
\begin{align*}
  L(\A) & \defeq \Sigma_\$^* \cdot L_{\text{acc}} & 
  L(\B) & \defeq L_\text{inv} \cap (\Sigma^*_\$ \cdot L_\text{acc}).
\end{align*}
Second, we let $g(n)\defeq 2^{f(n)}\cdot (f(n)+2)$. Consequently we
have
\begin{align*}
  & \# (L(\A) \cap \Sigma_\$^{g(n)}) - \# (L(\B) \cap \Sigma_\$^{g(n)}) > 0\\
  \iff & \M \text{ has at least one accepting run on } x.
\end{align*}
Hence, by keeping $\vec{p}$ defined as above with the redefined $g(n)$,
we have
\begin{align*}
  N(\A_h, \vec{p}) > N(\B_h, \vec{p}) \iff \M \text{ accepts } x.
\end{align*}
Note that even though $g(n)$ is exponential, its binary representation
requires only polynomially many bits.

Finally, we turn to the proof of the $\PEXP$ lower bound in
Proposition~\ref{prop:nfa-cfg-fixed-lower}(iii). To this end, let $\M$
be an $f(n)$-time bounded non-deterministic Turing machine where
$f(n)=2^{p(n)}$ for some polynomial $p(n)$. We could almost
straightforwardly reuse the construction given for
Proposition~\ref{prop:nfa-cfg-fixed-lower}(i) except that we cannot
construct an appropriate NFA recognising $L_{\text{inv}}$ in
logspace. The reason is that the working tape of $\M$ has
exponential length and we cannot uniquely determine a string of
exponential length with an NFA which, as stated above, depends on
$f(n)$. This can, however, be achieved by exploiting the exponential
succinctness of context-free grammars. More specifically, in
\cite{HuntRS76} it is shown that a CFG for the language
$L_{\text{inv}}$ can be constructed in logspace from the machine input
$x$. The $\PEXP$-hardness result of
Proposition~\ref{prop:nfa-cfg-fixed-lower}(iii) can then be shown
analogously to the \PP-hardness proof from
Proposition~\ref{prop:nfa-cfg-fixed-lower}(i) by encoding $\vec{p}$ in
binary.
\end{proof}

We prove the following  proposition from the main text:

\propDFAUnaryUpper*
\begin{proof}[Proof of Part~(i)]
  We show that the compressed word problem for unary DFA is in~\LOGSPACE.  Hence,
  \PIC\ is in~\LOGSPACE \ for unary DFA with Parikh vectors encoded in
  binary. Let $\A = (Q,\Sigma,q_0,F,\Delta)$ be the given unary
  DFA. W.l.o.g.\ we can assume that $Q = \{0, \ldots, m, m+1, \ldots,
  m+p-1\}$, where $q_0 = 0$, $i \xrightarrow{a} i+1$ for $0 \leq i <
  m+p-1$ and $m+p-1 \xrightarrow{a} m$. The numbers $m$ and $p$ can be
  computed in logspace by following the unique path of states from the
  initial state. For a given number $n$ encoded in binary we then have
  $a^n \in L(\A)$ if and only if $n \leq m$ and $n \in F$ or $n > m$
  and $((n-m) \bmod p) + m \in F$.  This condition can be checked in
  logspace, since all arithmetic operations on binary encoded numbers
  can be done in logspace (division is the most difficult one \cite{HAMB02}
  but note that here we only have to divide by a number $p$ with a
  logarithmic number of bits in the input size).
\end{proof}

\begin{proof}[Proof of Part~(ii)]
  Regarding hardness, one can reduce from the graph reachability
  problem, i.e., whether for a given directed graph $G = (V,E)$ there
  is a path from $s$ to $t$. By adding a loop at node $t$, this is
  equivalent to the existence of a path in $G$ from $s$ to $t$ of
  length $n=\#V$. Let $\A$ be the NFA obtained from $G$ by labeling
  every edge with the terminal symbol $a$ and making $s$ (resp., $t$)
  the initial (resp., unique final) state.  Moreover, let $\B$ be an
  NFA with $L(\B) = \emptyset$.  Then $N(\A,n) > N(\B,n)$ if and only
  if $a^n \in L(\A)$ if and only if there is a path in $G$ from $s$ to
  $t$ of length $n=|V|$.

  We now turn towards the upper bound. For $a, b \in \N$ we write $a +
  b \N$ for the set $\{a + b \cdot i : i \in \N\}$. Given a unary NFA
  $\A=(Q,\{a\},q_0,F,\Delta)$ with $p, q \in Q$ and $n \in \N$ we
  write $p \xrightarrow{n} q$ if there is a run of length~$n$ from $p$
  to~$q$. A simple algorithm follows from recent work by
  Sawa~\cite{Sawa13}:
  \begin{lemma}[{\cite[Lemma 3.1]{Sawa13}}] \label{lem-Sawa}
    Let $\A=(Q,\{a\},q_0,F,\Delta)$ be a unary NFA with $m \defeq |Q| \ge
    2$.  Let $n \ge m^2$.  Then $a^n \in L(\A)$ if and only if there are
    $q \in Q$, $q_f \in F$, $b \in \{1, \ldots, m\}$, and $a \in \{m^2
    - b - 1, \ldots, m^2 - 2\}$ with $n \in a + b \N$ and $q_0
    \xrightarrow{m-1} q \xrightarrow{b} q \xrightarrow{a - (m-1)}
    q_f$.
  \end{lemma}
  We use this lemma to show the following:

  \begin{proposition}\label{prop:nfa-unary-upper}
    The compressed word problem for unary NFA is in~\NL.
    Hence, \PIC\ is in~\NL\ for unary NFA with Parikh vectors encoded in binary.
  \end{proposition}
  \begin{proof}
    Let $\A=(Q,\{a\},q_0,F,\Delta)$ be the given unary NFA, and let
    $n \in \N$ be given in binary.  We claim that, given two states
    $p_1, p_2 \in Q$ and a number $c \in \N$ whose binary
    representation is of size logarithmic in the input size, we can
    check in~\NL\ whether $p_1 \xrightarrow{c} p_2$ holds.  To prove
    the claim, consider the directed graph~$G$ over vertices $Q \times
    \{0, \ldots, c\}$, with an edge from $(q_1, i)$ to $(q_2, j)$ if
    and only if $q_1 \xrightarrow{1} q_2$ and $j = i+1$.  The
    graph~$G$ can be computed by a logspace transducer.  Then $p_1
    \xrightarrow{c} p_2$ holds if and only if $(p_2, c)$ is reachable
    from $(p_1,0)$ in~$G$.  The claim follows as graph reachability is
    in~\NL.

    Now we give an \NL\ algorithm for the compressed word problem.  If
    $n < m^2$ then guess $q_f \in F$ and check, using the claim above,
    in~\NL\ whether $q_0 \xrightarrow{n} q_f$.  If $n \ge m^2$ we use
    Lemma~\ref{lem-Sawa} as follows. We run over all $q \in Q$, $q_f
    \in F$, $b \in \{1, \ldots, m\}$, and $a \in \{m^2 - b - 1,
    \ldots, m^2 - 2\}$ (all four values can be stored in logspace),
    and check (i) whether $n \in a + b \N$ and (ii) $q_0
    \xrightarrow{m-1} q \xrightarrow{b} q \xrightarrow{a - (m-1)} q_f$
    holds.  Condition (i) can be checked in logspace (as in the proof
    of Proposition~\ref{prop:dfa-unary-upper}), and condition (ii) can
    be checked in ~\NL\ by the above claim.
  \end{proof}

  It follows that \PIC\ is in~\NL\ for unary NFA with Parikh vectors
  encoded in binary: Given NFA $\A,\B$ and $n \in \N$ in binary, we
  have $N(\A,n) > N(\B,n)$ (where we identity the mapping $\vec{p} : \{a\} \to \N$
  with the single number $\vec{p}(a)$) if and only if $N(\A,n) = 1$ and $N(\B,n) =
  0$, which holds if and only if $a^n \in L(\A)$ and $a^n \not\in L(\B)$.
  Since $\NL$ is closed under complement, the latter condition can be
  checked in $\NL$.
\end{proof}

\begin{proof}[Proof of Part~(iii)]
  The \P-upper bound is clear since the word problem for CFG is in \P.
  Regarding the \DP-upper bound, Huynh~\cite{Huy84} shows that the
  uniform word problem for context-free grammars over a singleton
  alphabet $\{a\}$ is \NP-complete; where the input word $a^n$
  is given by the binary representation of $n$. Given CFG $\A,\B$ over
  $\{a\}$ and a number $n$, we have that
  $N(\A,n) > N(\B,n)$ if and only if $a^n\in L(\A)$ and
  $a^n\not \in L(\B)$. The latter is a decision problem in \DP\ by the
  above result from \cite{Huy84}.
  
  Let us now turn to the lower bounds. For the \P-lower bound, we reduce
  from the $\epsilon$-membership problem for context-free grammars,
  which is known to be \P-hard even for grammars with no terminal
  symbol~\cite[Prob.\ A.7.2]{GHR95}. Let $\A$ be such a grammar, and
  let $\B$ be such that $L(\B)=\emptyset$. Then $N(\A,0) > N(\B,0)$ if
  and only if $\epsilon \in L(\A)$. For the \DP-lower bound, the
  following problem is \DP-complete: Given two instances $(s, v_1,
  \ldots, v_m)$, $(t, w_1, \ldots, w_n)$ of \textsc{SubSetsum} (all
  numbers $s, v_1, \ldots, v_m,t, w_1, \ldots, w_n$ are binary
  encoded), does the following hold?
  \begin{itemize}
  \item There exist $a_1, \ldots, a_m \in \{0,1\}$ with $s = a_1 v_1 + \cdots + a_m v_m$.
  \item For all $b_1, \ldots, b_n \in \{0,1\}$, $t \neq b_1 w_1 +
    \cdots + b_n w_n$.
  \end{itemize}
  $\DP$-hardness of this problem follows by a reduction from the
  problem \textsc{SAT-UNSAT} (one can take the standard reduction from
  \textsc{SAT} to \textsc{SubsetSum}). Let us assume that $s \geq t$
  (if $t > s$ we can argue similarly).  We then construct in logspace
  CFG $\A$ and $\B$ such that $\A$ produces all words of the form
  $a^{a_1 v_1 + \cdots + a_m v_m}$, where $a_1, \ldots, a_m \in
  \{0,1\}$ and $\B$ produces all words of the form $a^{(s-t) + b_1 w_1
    +\cdots+ b_n w_n}$, where $b_1, \ldots, b_m \in \{0,1\}$. We then
  have $a^s \in L(\A) \setminus L(\B)$ if and only if there are $a_1,
  \ldots, a_m \in \{0,1\}$ with $s = a_1 v_1 + \cdots + a_m v_m$ but
  $t \neq b_1 w_1 + \cdots + b_n w_n$ for all $b_1, \ldots, b_n \in
  \{0,1\}$.
\end{proof}

}{}

\end{document}